\documentclass[a4paper,11pt]{article}
\usepackage{latexsym,amssymb,amsfonts,amsmath,amsthm}
\usepackage[dvips]{graphicx}

\newtheorem{theorem}{Theorem}
\newtheorem{lemma}{Lemma}
\newtheorem{corollary}{Corollary}
\newtheorem{proposition}{Proposition}
\newtheorem{remark}{Remark}
\newtheorem{definition}{Definition}
\newtheorem{example}{Example}

\newcommand{\bs}[1]{\boldsymbol{#1}}

\newcommand{\spann}{{\rm span}}

\usepackage{soul}
\usepackage{color}
\newcommand{\rp}[1]{{\color{black} #1}}%

\title{{\Large {\bf Partition-based discrete-time quantum walks  
}
}}
\author{ 
Norio Konno,$^{1}$ 
\footnote{konno@ynu.ac.jp 
}\quad
Renato Portugal,$^{2}$ 
\footnote{
portugal@lncc.br 
}\quad
Iwao Sato,$^{3}$ 
\footnote{
isato@oyama-ct.ac.jp 
}\quad 
Etsuo Segawa$^{4}$ 
\footnote{e-segawa@m.tohoku.ac.jp (corresponding author)
}
\\ 
{\scriptsize $^1$ 
 Department of Applied Mathematics, Faculty of Engineering, Yokohama National University, 
}\\
{\scriptsize 
Hodogaya, Yokohama 240-8501, Japan. 
} \\
{\scriptsize $^2$ 
National Laboratory of Scientific Computing - LNCC,
}\\
{\scriptsize 
Av. Get{\'u}lio Vargas 333, Petr{\'o}polis, RJ, 25651-075, Brazil
} \\
{\scriptsize $^3$ 
Oyama National College of Technology,
}\\
{\scriptsize
 Oyama, Tochigi 323-0806, Japan
}\\
{\scriptsize $^{4}$ 
Graduate School of Information Sciences, Tohoku University,
}\\
{\scriptsize 
Aoba, Sendai 980-8579, Japan
} \\
} 
\date{\empty }
\begin{document}
\maketitle

\par\noindent
\begin{small}
\par\noindent
{\bf Abstract}. 
\rp{We introduce a family of discrete-time quantum walks, called two-partition model, based on two equivalence-class partitions of the computational basis, which establish the notion of local dynamics. This family encompasses most versions of unitary discrete-time quantum walks driven by two local operators studied in literature, such as the coined model, Szegedy's model, and the 2-tessellable staggered model. We also analyze the connection of those models with the two-step coined model, which is driven by the square of the evolution operator of the standard discrete-time coined walk. We prove formally that the two-step coined model, an extension of Szegedy model for multigraphs, and the two-tessellable staggered model are unitarily equivalent. Then, selecting one specific model among those families is a matter of taste not generality.}

\end{small}

\setcounter{equation}{0}

\section{Introduction}

The quantum walk is a quantized version of the classical random walk. The discrete-time version can be obtained from the path-integral formulation of quantum mechanics~\cite{FH}, which was addressed for instance in Refs.~\cite{Am,CIR} for the infinite line. The most-studied discrete-time version on graphs was proposed in Ref.~\cite{Aharonov:2000} and is known as the coined model because the walker must have an internal state, which is used to determine the direction of the step. The coin is not mandatory, in fact, neither the Szegedy quantum walk~\cite{Sze} nor the staggered model~\cite{Por1} has a coin operator. These latter models define partitions of the vertex set in order to establish the model's evolution operator.
The quantum walk offers a good opportunity for experimental implementations~(see \cite{Jingbo} and references therein) and is an interesting model for analyzing topological phases~\cite{Kitagawa}.

Quantum walks are discussed from many viewpoints as an interdisciplinary research field. From the pure-mathematics viewpoint, the quantum random walk was discussed in the area of quantum probability~\cite{Gudder,Par88} and, more recently, Refs.~\cite{HKSS,HKSS2} introduced the notion of quantum-graph walks. Ref.~\cite{MOS} proposed an extension of quantum walks to simplicial complexes, Ref.~\cite{CGMV} used CMV matrices, proposed in~\cite{CMV} for studies of orthogonal Laurent polynomials on the unit circle, in the analysis of quantum walks, and Ref.~\cite{Konnobook} obtained  limit theorems for quantum walks on the line.
From the scattering viewpoint, quantum walks can be seen as waves that are transmitted and reflected at each vertex~\cite{FH04}. From the computer-science viewpoint, quantum walks can be used to detect and to find marked vertices faster than classical random walks~\cite{AKR,Sze,SKW}.


In this paper, we introduce the notion of partition-based quantum walk with the goal of analyzing the equivalence of quantum walk models under a common framework. We address four models: The coined, Szegedy, 2-tessellable staggered quantum walks, and a new model called two-partition quantum walk, which is a partition-based quantum walk defined by two independent partitions of the computational basis. The partition elements establish the notion of locality or neighborhood. We prove that the coined, Szegedy, and 2-tessellable staggered models are two-partition quantum walks. We also address the converse statement. In order to show that the two-partition model is contained in the Szegedy and 2-tessellable staggered models, we have to extend the Szegedy model for multigraphs and we have to loosen the way one chooses the local unitary operators. Notice that, as a corollary, we obtain that the coined model is included in the extended Szegedy and 2-tessellable staggered models. Those results generalize the analysis of Refs.~\cite{Por4,PS}.

The two-partition model is not contained in the coined model even extending the shift operator. It is well known that the Szegedy quantum walks can be included in the two-step coined model, which employs an evolution operator that is the square of the evolution operator of the coined model, by using the swap operator as the shift operator~\cite{LW17}. This motivates us to analyze the two-step coined model. We are able to prove that the two-partition model is included in the two-step coined model. Since the two-step coined model is a two-partition model, we prove that those models are unitarily equivalent (see Lemma~\ref{lem}). Our results show that the two-step coined model and the extended versions of the Szegedy and 2-tessellable models are unitarily equivalent~(see Theorem~\ref{unitaryeq}).

In order to establish a unitary equivalence of the evolution operators of the quantum walk models, we need to give a precise interpretation of the mathematical description of the walker's allowed locations for each model. In the coined model, the walker steps on the arcs of the graph. In the Szegedy model, the walker steps on the edges of the graph, and in the staggered model, the walker steps on the vertices of the graph. In the original coined model, it is possible to give a precise direction to the walker's steps via the shift operator. In the Szegedy and 2-tessellable staggered models, the evolution operator is the product of two local unitary operators and, under the action of each local operator, the walker goes to more than one location, using the state superposition principle of quantum mechanics. The coined model and the extended version of the Szegedy model use multigraphs. The staggered model always uses simple graphs.

This paper is organized as follows. 
In section 2, we define the two-partition quantum walk. 
In section 3, we show that the coined, Szegedy, and 2-tessellable quantum walks are two-partition quantum walks. We also define quantum walks on hypergraphs and show that it is also a two-partition quantum walk.
In section 4, we address the unitary equivalence among the models and prove Theorem~\ref{unitaryeq}, which is the main result of this work. 
Finally, in section 5, we perform the spectral analysis of coined walks.

\section{Two-partition quantum walk}
Let $\Omega=\{\omega_1,\omega_2,\dots\}$ be a countable set. 
We define two decompositions of $\Omega$ induced by equivalence relations $\pi_1$ and $\pi_2$ over $\Omega$, such that 
	\begin{align*} 
        \Omega/\pi_1 &= \{ [\omega]_{\pi_1}\;|\; \omega\in \Omega \}, \\
        \Omega/\pi_2 &= \{ [\omega]_{\pi_2}\;|\; \omega\in \Omega \},
        \end{align*}
where $[\omega]_{\pi_j}=\{ \omega'\in \Omega \;|\; \omega'\stackrel{\pi_j}{\sim} \omega\}$ ($j=1,2$) is the equivalence class of $\omega$.
Let $J_i$ be the cardinality of the quotient set $\Omega/\pi_i$ ($i=1,2$) and 
let $[J_i]$ denote the set $\{1,2,\dots,J_i\}$ ($i\in 1,2$).
We call the elements of $\Omega/\pi_1$ by $C_i$ for $i\in [J_1]$ and the elements of $\Omega/\pi_2$ by $D_j$ for $j\in [J_2]$ 
and we assume that $|C_i|<\infty$ and $|D_j|<\infty$.

The Hilbert space induced by $\Omega$ is defined as
	\[\mathcal{H}=\ell^2(\Omega)=\{\psi: \Omega\to\mathbb{C} \;|\; \sum_{\omega\in \Omega}|\psi(\omega)|^2<\infty \}, \] 
whose inner product is the standard one. Each equivalence relation $\pi_j$ induces an orthogonal decomposition of $\mathcal{H}$ as follows
	\[ \mathcal{H}=\bigoplus_{i\in [J_1]} \mathcal{C}_i=\bigoplus_{j\in [J_2]} \mathcal{D}_j, \]
where
	\begin{align*} 
        \mathcal{C}_i &=\{\psi\in \mathcal{H} \;|\; \omega\notin C_i \Rightarrow \psi(\omega)=0 \},\;\;(i\in [J_1]), \\
	\mathcal{D}_j &=\{\psi\in \mathcal{H} \;|\; \omega\notin D_j \Rightarrow \psi(\omega)=0 \},\;\;(j\in [J_2]), 
	\end{align*}
that is, $\mathcal{C}_i=\spann\{ \delta_\omega \;|\; \omega\in C_i \}$ and $\mathcal{D}_j=\spann\{ \delta_\omega \;|\; \omega\in D_j \}$,
where $\delta_\omega(\cdot)$ is the Kronecker delta function 
	\[ \delta_\omega(\omega')=\begin{cases} 1 & \text{if $\omega=\omega'$,} \\ 0 & \text{ otherwise.} \end{cases} \]
Let $\hat{E}_i$ and $\hat{F}_j$ be local unitary operators on $\mathcal{C}_i$ and $\mathcal{D}_j$ for each $i\in [J_1]$ and $j\in [J_2]$, respectively. 
$\hat{E}_i$ is a local operator on $\mathcal{C}_i$
when $\langle \delta_{\omega'}, \hat{E}_i\delta_\omega\rangle =0$  if $\omega$ or $\omega'$ does not belong to $C_i$, 
and, in the same way, 
 $\hat{F}_j$ is a local operator on $\mathcal{D}_j$
when $\langle \delta_{\omega'}, \hat{F}_j\delta_\omega\rangle =0$ if $\omega$ or $\omega'$ does not belong to $D_j$.  
We set $\hat{E}=\oplus_{i\in [J_1]}\hat{E}_i$, $\hat{F}=\oplus_{j\in [J_2]}\hat{F}_j$. 
Thus, $\langle \delta_{\omega'},\hat{E}\delta_{\omega}\rangle=0$ if $\omega$ and $\omega'$ belong to different partition elements of $\pi_1$, and also 
$\langle \delta_{\omega'},\hat{F}\delta_{\omega}\rangle=0$ if $\omega$ and $\omega'$ belong to different partition elements of $\pi_2$.
\begin{definition}
The two-partition walk $(\Omega;\pi_1,\pi_2;\hat{U})$ is defined by the following items:
\begin{enumerate}
\item The associated Hilbert space is $\mathcal{H}=\ell^2(\Omega)$ endowed with the standard inner product. 
\item The evolution operator on $\mathcal{H}$ is $\hat{U}(\Omega,\pi_1,\pi_2;\{\hat{E}_i\}_{i\in [J_1]},\{\hat{F}_j\}_{j\in [J_2]}):=\hat{F}\, \hat{E}$ and denoted simply by $\hat{U}$.
\item The probability distribution is $\mu_n^{(\psi_0)}: \Omega \to [0,1]$ for $n\in \mathbb{N}$, $\psi_0\in \mathcal{H}$ such that
 	\[ \mu_n^{(\psi_0)}(\omega)=|\hat{U}^n\psi_0(\omega)|^2, \]
or $\mu_n^{(\psi_0)}(C_j) := \sum_{\omega\in C_j}\mu_n^{(\psi_0)}(\omega)$.         
\end{enumerate}
\end{definition}
The two-partition walk is a quantum walk on the discrete set $\Omega$ with an evolution operator $\hat{U}$ that depends on partitions $\pi_1$ and $\pi_2$, 
which provide the notion of locality in $\Omega$. 
The partitions allow us to choose two block diagonal unitary operators $\hat{E}$ and $\hat{F}$, which determine the evolution operator through the expression $\hat{U}=\hat{F}\, \hat{E}$. 
We use $\hat{U}$ as representing the whole framework of the two-partition walk $(\Omega,\pi_1,\pi_2;\hat{U})$. 
\begin{example}\label{ex1}
Let $\Omega=\{(a,b),(a,c),(d,c)\}$ and let partitions $\pi_1$ and $\pi_2$ be respectively defined by
\begin{align*}  
(\alpha,\beta) \stackrel{\pi_1}{\sim} (\alpha',\beta') \,\,{\Longleftrightarrow}\,\, \beta=\beta',\\ 
(\alpha,\beta) \stackrel{\pi_2}{\sim} (\alpha',\beta') \,\,{\Longleftrightarrow}\,\, \alpha=\alpha'. 
\end{align*} 
In this setting, 
we have $C_1=\{(a,b)\}$, $C_2=\{(a,c),(d,c)\}$, $D_1=\{(a,b),(a,c)\}$, $D_2=\{(d,c)\}$ and 
the matrix expression for the first and second operators $\hat{E}$ and $\hat{F}$ are given by 
\[
\hat{E}=\left[
	\begin{tabular}{c|cc}
        $*$ & $0$ & $0$ \\ \hline
        $0$ & $*$ & $*$ \\ 
        $0$ & $*$ & $*$    
        \end{tabular}
        \right], \quad
\hat{F}=\left[
	\begin{tabular}{cc|c}
        $*$ & $*$ & $0$ \\
        $*$ & $*$ & $0$ \\ \hline
        $0$ & $0$ & $*$    
        \end{tabular}
        \right] 
\]
in the standard orthogonal basis of $\mathbb{C}^3$, 
$\delta_{(a,b)}\cong {}[1,0,0]^T$, $\delta_{(a,c)}\cong {}[0,1,0]^T$, $\delta_{(d,c)}\cong {}[0,0,1]^T$, and the entries $*$ can be nonzero.
\end{example}
\section{Examples of two-partition walk}
\subsection{Bipartite walk}
The bipartite walk is defined on the edge set $E$ of a bipartite multigraph $G=(X\sqcup Y, E)$, 
where $X \sqcup Y$ is the disjoint union of sets $X$ and $Y$. 
The multigraph is connected and, as a particular case, can be a bipartite simple graph.
The $X$-end point of $e\in E$ is denoted by $X(e)$ and the other end point is denoted by $Y(e)$. 
Setting $\Omega=E$, we define equivalence relations $\pi_1$ and $\pi_2$ on $E$ by 
	\begin{align*}
        e \overset{\pi_1}{\sim} f & \,\,{\Longleftrightarrow}\,\, X(e)=X(f),  \\  
        e \overset{\pi_2}{\sim} f & \,\,{\Longleftrightarrow}\,\, Y(e)=Y(f). 
        \end{align*} 
The equivalence relation $\pi_1$ provides a partition of $E$ into equivalence classes $[e]_{\pi_1}=\{f\in E \;|\; f\stackrel{\pi_1}{\sim}e \}$ and, likewise, $\pi_2$ provides a partition into $[e]_{\pi_2}=\{f\in E \;|\; f\stackrel{\pi_2}{\sim}e \}$. The respective quotient sets are
	\begin{align*}
        \Omega/\pi_1 &=\{[e]_{\pi_1} \;|\; e\in E\}=\{ C_x \;|\; x\in X \}\cong X, \\
        \Omega/\pi_2 &=\{[e]_{\pi_2} \;|\; e\in E\}=\{ D_y \;|\; y\in Y \}\cong Y,
        \end{align*}
where $C_x=\{e\in E\;|\;X(e)=x\}$ and $D_y=\{e\in E\;|\;Y(e)=y\}$. 

The one-step dynamics of this walk from an initial edge $e\in E$ is as follows: 
In the first half step under the action of the unitary operator $\hat{E}$, the walker moves from $e$ to a neighbor edge $f$ so that $e$ and $f$ share a common end vertex in $X$, that is, $X(e)=X(f)$.
In the second half step under the action of the unitary operator $\hat{F}$, the walker moves from $f$ to a neighbor edge $g$ so that $f$ and $g$ share a common end vertex in $Y$, that is, $Y(f)=Y(g)$.
The bipartite walk is determined by a bipartite multigraph $G$ and by an evolution operator $\hat{W}$, which is the product of two local unitary operators each one obtained from the direct-sum of $\{\hat{R}_x\}_{x\in X}$ and $\{\hat{R}_y\}_{y\in Y}$, respectively. 
The bipartite walk is described by $(G; \hat{W})$ or simply by $\hat{W}$. 

The Szegedy model~\cite{Sze} is a subclass of the class of bipartite walks. A bipartite walk is an instance of the Szegedy model if the multigraph $G$ is a simple bipartite graph (no multiple edges) and the local unitary operators are obtained from stochastic matrices associated with a classical Markov chain, as described in Ref.~\cite{Sze}.

\noindent\\
\noindent{\bf Extending the Szegedy model }\\
Now we present an extension of the Szegedy model for bipartite multigraphs using the bipartite walk. Consider a classical Markov chain defined on a connected multigraph $H=(V,E)$. Since the bipartite walk is defined on a bipartite multigraph, we consider the duplication of the original multigraph $H$ similar to the method used by Szegedy. 
The duplicated multigraph is the bipartite multigraph $G=(V_2,E_2)$, where $V_2=V\cup V'$,
$V'$ is a copy of $V$, that is, $V'=\{v' \;|\; v\in V\}$, and each edge $(u,v)\in E$ corresponds to an edge $(u,v')\in E_2$ so that $|E(u,v)|=|E_2(u,v')|$,
where $E(u,v)$ is the set of edges in $E$ whose end vertices are $u$ and $v$.

Consider two functions $p: E(G) \rightarrow [0,1]$ and $q: E(G) \rightarrow [0,1]$ with 
\begin{equation}\label{liftupcondition}
\sum_{V(e)=x} p(e)= \sum_{V'(e)=y} q(e)=1, \,\,\forall x\in V , \,\,\forall y \in V',   
\end{equation}
so that 
the original Markov chain on $H$ is naturally lifted up to this duplicated multigraph $G$ by demanding that, for all $u,v\in V$,
\[
\{p(e) \;|\; e\in E_2(u,v')\} = \{q(e) \;|\; e\in E_2(u',v)\}.
\]

In the framework of two-partition walks, we set $\Omega=E_M$ and define $\pi_1$ and $\pi_2$ as 
\begin{align*}
e\stackrel{\pi_1}{\sim}f &\,\,\Longleftrightarrow\,\, V(e)=V(f),\\
e\stackrel{\pi_2}{\sim}f &\,\,\Longleftrightarrow\,\, V'(e)=V'(f).
\end{align*}  
We assign the local unitary operators $\{\hat{R}_x\}_{x\in V}$ and $\{\hat{R}_y\}_{y\in V'}$ on each vector spaces $\mathcal{C}_x$ and $\mathcal{D}_y$ using the following formulas
	\begin{align*}
        \hat{R}_x &= 2|\alpha_x\rangle \langle \alpha_x|-\bs{1}_{\mathcal{C}_x}, \\
        \hat{R}_y &= 2|\beta_y\rangle \langle \beta_y|-\bs{1}_{\mathcal{D}_y},
        \end{align*}
where ``$|\gamma \rangle\langle \gamma |$" represents the orthogonal projection operator onto $\gamma\in \mathcal{H}$, and 
$|\alpha_x\rangle$ and $|\beta_y\rangle$ belong to $\mathcal{C}_x$ and $\mathcal{D}_y$, respectively, defined by 
	\begin{align*} 
        \alpha_x(e) &= \begin{cases} \sqrt{p(e)} & \text{if $V(e)=x$,} \\ 0 & \text{otherwise,} \end{cases}\\ 
        \beta_y(e)  &= \begin{cases} \sqrt{q(e)} & \text{if $V'(e)=y$,} \\ 0 & \text{otherwise.}  \end{cases}
        \end{align*}
The evolution operator is $\hat{W}=\left(\oplus_{y\in V'} \hat{R}_y\right)\left(\oplus_{x\in V} \hat{R}_x\right)$.

\noindent\\
\noindent{\bf Quantum search in the extended Szegedy model }\\
Suppose we define a classical Markov chain in a connected multigraph $H=(V,E)$. 
Searching a vertex in $H$ employing the Markov chain is accomplished as follows: The marked vertices are converted into sinks (or absorbing vertices), by removing the arcs outgoing from the marked vertices. This procedure
generates a new directed multigraph that we call $H_M=(V_M,E_M)$, where $M$ is the set of marked vertices.
The same procedure is used in the extended Szegedy model. The bipartite graph $G=(V_2,E_2)$ is also
converted into a directed bipartite multigraph $G_M=(V_2^M,A_2^M)$, which is called modified multigraph. 
The first column of Fig.~\ref{fig:graphM} depicts an example of a multigraph $H$ and its version with one marked vertex represented as an empty vertex. The second column depicts the corresponding bipartite versions, on which the extended Szegedy quantum walk takes place.

\begin{figure}[htbp]
  \begin{center}
           \includegraphics[clip, width=7cm]{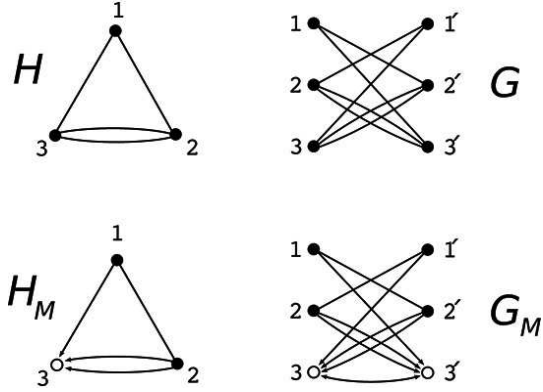}
              \caption{Example of a multigraph $H$ and its corresponding bipartite multigraph $G$ (duplicated multigraph). If $M=\{3\}$, $H_M$ is a directed multigraph and $G_M$ is its corresponding directed bipartite multigraph, which has edges linking each marked vertex and its copy. A marked vertex is a sink and is represented by an empty vertex. }
  \end{center}
  \label{fig:graphM}
\end{figure}

Note that it is possible to preserve completely the above classical dynamics on the directed multigraph $G_M$ with sinks 
even on the non-directed bipartite multigraph, whose edge set is the support of the arc set of $G_M$ when we modify the 
transition probability in the following way. 
Let $E_3$ be the set of non-directed edges linking each marked vertex with its copy. 
The modified $p'$ and $q'$ with domain $E_2^M=\{ |a| \;|\; a\in A_2^{M} \}$ are given by
	\begin{align*}
        p'(e) &=
        \begin{cases}
        p(e) & \text{if $V(e)\notin M$,} \\
        1 & \text{if $e\in E_3$,} \\
        0 & \text{otherwise} 
        \end{cases}
        \\
        q'(e) &=
        \begin{cases}
        q(e) & \text{if $V'(e)\notin M$,} \\
        1 & \text{if $e\in E_3$,} \\
        0 & \text{otherwise.}
        \end{cases}
        \end{align*}  
Since bipartite walks are defined on non-directed bipartite multigraphs,
the quantum-walk dynamics on the modified multigraph is readily obtained from the extended Szegedy model 
as soon as we describe the modified stochastic transition matrix: 
	\[ (P')_{u,v}=  
        \begin{cases} \sum_{u=V(e),v=V'(e)} p'(e)  & \text{: $u\in V, v\in V'$,} \\ \sum_{u=V'(e),v=V(e)} q'(e)  & \text{: $u\in V', v\in V$.}\end{cases}\]

The dynamics is driven by $\hat{W}=\left(\oplus_{y\in V'} \hat{R}_y'\right)\left(\oplus_{x\in V} \hat{R}_x'\right)$ on $\ell^2(E_2^{M})$, where $\hat{R}_x'$ and $\hat{R}_y'$ are defined in terms of $\alpha_{x}'(e)$ and $\beta_{y}'(e)$
given by
\begin{align*} 
        \alpha_x'(e) &= \begin{cases} \sqrt{p'(e)} & \text{if $V(e)=x$,} \\ 0 & \text{otherwise,} \end{cases}\\ 
        \beta_y'(e)  &= \begin{cases} \sqrt{q'(e)} & \text{if $V'(e)=y$,} \\ 0 & \text{otherwise.}  \end{cases}
\end{align*}
Notice that $\alpha_{o(a)}'(|a|)=0$ for the edge $|a|$ corresponding to removing arc $a$ with $o(a)\in M$ in $G_M$, as well as 
$\alpha_{o(b)}'(|b|)=0$ for the edge $|b|$ corresponding to removing arc $b$ with $o(b)\in M'$ in $G_M$.

We consider $\hat{W}\delta_e$ with $e\in E_3$, 
since $\alpha_{V(e)}=\delta_e$ and $\beta_{V'(e)}=\delta_e$,  
we have $\hat{R}_{V(e)}\delta_e=\delta_e$ and $\hat{R}_{V'(e)}'\delta_e=\delta_e$. 
Therefore $\hat{W}\delta_e=\delta_e$ holds, which implies that $\hat{W}$ acts as the identity operator on the subspace spanned by $\{ \delta_e \;|\; e\in E_3 \}$. 
From this above observation, under the decomposition $\ell^2(E_2^{M})=\ell^2(E_2)\oplus \ell^2(E_3)$, $\hat{R}_X:=\oplus_{x\in X}\hat{R}_x$ and $\hat{R}_Y:=\oplus_{y\in Y}\hat{R}_y'$ are reexpressed by
$\hat{R}_X=\hat{R}_{X,M}\oplus \bs{1}_{E_3}$, $\hat{R}_Y=\hat{R}_{Y,M}\oplus \bs{1}_{E_3}$, where 
	\begin{align*}
        \hat{R}_{X,M} &=  2\sum_{x\in X} |\alpha_{M,x}\rangle \langle \alpha_{M,x}|-\bs{1}_{E_2}, \\
        \hat{R}_{Y,M} &=  2\sum_{x\in X} |\beta_{M,x}\rangle \langle \beta_{M,x}|-\bs{1}_{E_2}. 
        \end{align*} 
Here
$\alpha_{x,M},\beta_{y,M}\in \ell^2(E_M)$ are the cut off on the marked elements of $\alpha_x'$ and $\beta_y'$, for $x\in V$, $y\in V'$, respectively, defined by
	\begin{align}\label{peropero}
        \alpha_{x,M}(e)= \begin{cases}  \alpha'_x(e) & \text{if $x\notin M$,} \\ 0 & \text{otherwise,} \end{cases}\,\,\textrm{and}\,\,
        \beta_{y,M}(e) = \begin{cases}  \beta'_y(e) & \text{if $y\notin M'$,} \\ 0 & \text{otherwise.} \end{cases}
        \end{align}
Then, the bipartite walk with the quantum search is expressed by
	\[ \hat{W}=\hat{R}_Y\hat{R}_X=\hat{R}_{Y,M}\hat{R}_{X,M}\oplus \bs{1}_{E_3} \]
under the decomposition $\ell^2(E_M)=\ell^2(E_2)\oplus \ell^2(E_3)$. 
Since the initial state is usually given by
	\[ \psi_0=\frac{1}{\sqrt{|V|}} \sum_{e\in E_2}\sqrt{p(e)}\delta_e\in\ell^2(E_2), \]
which has no overlap to the eigenspace spanned by $E_3$, 
we can concentrate on the main operator $\hat{R}_{Y,M}\hat{R}_{X,M}$ on the subspace generated by $E_2$. 

Thus, the evolution operator of the extended Szegedy model for $M\cup M'\subset V_2$ driven by a bipartite walk can be reduced to the following settings: 
	\begin{align} 
	\Omega & =E_2, \notag \\
	e \stackrel{\pi_1}{\sim} f &\Leftrightarrow V(e)=V(f), \label{revisedSzeMethod1}\\
	e \stackrel{\pi_2}{\sim} f &\Leftrightarrow V'(e)=V'(f). \notag
	\end{align}
The evolution operator $\hat{W}=\hat{F}\hat{E}$ on $G_2$ with $\hat{F}=\oplus_{y\in V'}\hat{F}_y$ and $\hat{E}=\oplus_{x\in V}\hat{E}_x$ is expressed by
	\begin{align}
	\hat{E}_x &= 2|\tilde{\alpha}_x\rangle\langle \tilde{\alpha}_x|-\bs{1}_{\mathcal{C}_x},\notag \\
	\hat{F}_y &= 2|\tilde{\beta}_y\rangle\langle \tilde{\beta}_y|-\bs{1}_{\mathcal{D}_x}, \label{revisedSzeMethod2}
	\end{align}
where $\tilde{\alpha}_x,\tilde{\beta}_{y}\in \ell^2(E_2)$ satisfy 
\begin{enumerate}
\item\label{pipi1} $\tilde{\alpha}_x(e)=\tilde{\beta}_{x'}(f)$ for every $e,f\in E_2$ and 
\[ (V(e))'=V'(f),\; (V(f))'=V'(e), \] 
\item\label{pipi2} if $V(e)\neq x$ then $\tilde{\alpha}_x(e)=0$, 
\item\label{pipi3}  for every $x\in V$ 
\[ ||\tilde{\alpha}_x||=\begin{cases} 1 & \text{if $x\notin M$,} \\ 0 & \text{otherwise.}\end{cases} \]
\end{enumerate}
Conditions (\ref{pipi1}) and (\ref{pipi2}) are generalization of (\ref{liftupcondition}). 
Condition (\ref{pipi3}) is equivalent to (\ref{peropero}).
From now on, we can regard $\hat{R}_{Y,M}\hat{R}_{X,M}$ on $E_2$ as the evolution operator of the extended Szegedy model. 
\subsection{Coined walk}
The coined walk is determined by a multigraph $G=(V, A)$, 
where $A$ is the set of symmetric arcs induced by edge set of $G$, that is, $a\in A$ if and only if $\bar{a}\in A$, where
$\bar{a}$ is the inverse arc of $a$. 
The origin of $a$ is denoted by $o(a)\in V$ and the terminus of $a$ is denoted by $t(a)$. 
For $a\in A$, $|a|$ is the edge in $E$ inducing $a$. 
Setting $\Omega=A$, we define the following equivalence relations
	\begin{align*}
        a \overset{\pi_1}{\sim} b & \,\,{\Longleftrightarrow}\,\, t(a)=t(b),  \\ 
        a \overset{\pi_2}{\sim} b & \,\,{\Longleftrightarrow}\,\, |a|=|b|. 
        \end{align*} 
The equivalence relation $\pi_1$ provides a partition of $A$ into equivalence classes 
$[a]_{\pi_1}=\{b\in A \;|\; b\stackrel{\pi_1}{\sim}a \}$ and, likewise, $\pi_2$ provides a partition into $[a]_{\pi_2}=\{b\in A \;|\; b\stackrel{\pi_2}{\sim}a \}$. 
The respective quotient sets are 
	\begin{align*}
        \Omega/\pi_1=\{[a]_{\pi_1}\;|\;a\in A\}=\{ C_u \;|\;u\in V \}\cong V, \\
        \Omega/\pi_2=\{[a]_{\pi_2}\;|\;a\in A\}=\{ D_e \;|\;e\in E \}\cong E.
        \end{align*}
We set $C_u:=\{a\in A\;|\;t(a)=u\}$ and $D_e:=\{a\in A\;|\;|a|=e\}$. 
The unitary operator $\hat{E}=\oplus_{u\in V}\hat{E}_u$ associated with $C_u$ is called the coin operator. 
The unitary operator $\hat{F}=\oplus_{e\in E}\hat{F}_e$ associated with $D_e$ is usually defined as 
$\hat{F}_{|a|}\delta_a=\delta_{\bar{a}}$ and  $\hat{F}_{|a|}\delta_{\bar{a}}=\delta_{a}$, and 
is called the flip-flop shift operator. 
The one-step dynamics of this walk from an initial arc $a\in A$ is as follows: 
At the first half step under the action of the coin operator $\hat{E}$, a walker on the arc $a$ moves to a neighbor arc $b$ that has a common terminal vertex, that is, $t(a)=t(b)$. 
At the second half step under the action of the flip-flop shift operator $\hat{F}$, the walker on $b$ flips the direction to $\bar{b}$. 
One time-step can be regarded as the dynamics of a plane wave, which is reflected and transmitted in every vertex and its relation to the quantum graph is addressed in Ref.~\cite{HKSS}.
The evolution operator is $\hat{\Gamma}=\hat{F}\,\hat{E}$. 

\

\noindent{\bf Extending the shift operator } \\
A natural extension of the coined walk is to extend the ``shift" operator $\hat{F}_e$ corresponding to the transposition
so that $\hat{F}_{e}$ is a general two-dimensional unitary operator. 
When we perform such an extension, we can find the unitary matrices 
in the studies of the CMV matrix~\cite{CMV}, a radio activity isotope separation by alternative terahertz pulse engineering~\cite{MY}, 
and quantum simulation of topological phases~\cite{Kitagawa,Asboth}. 
For example, for the CMV matrix, the corresponding coined walk with the extended shift operator is expressed as follows: 
the graph is the one-dimensional half integer lattice, and the coin and extended shift operators are 
\[ C=1\oplus \begin{bmatrix} \bar{\gamma}_1 & \rho_1 \\ \rho_1 & -\gamma_1 \end{bmatrix} 
      \oplus \begin{bmatrix} \bar{\gamma}_3 & \rho_3 \\ \rho_3 & -\gamma_3 \end{bmatrix}\oplus \cdots \]
\[ S=\begin{bmatrix} \bar{\gamma}_0 & \rho_0 \\ \rho_0 & -\gamma_0 \end{bmatrix} 
      \oplus \begin{bmatrix} \bar{\gamma}_2 & \rho_2 \\ \rho_2 & -\gamma_2 \end{bmatrix}\oplus \cdots\]      
under the order of the standard basis of the coined walk $(0;-),(1;+),(1;-),(2;+),(2;-),\dots$, 
where $(i;\epsilon)$ is the arc of the half integer whose terminus is $i$ and origin is $i-\epsilon$. 
Here $\gamma_j\in\mathbb{C}$ with $|\gamma_j|\leq 1$ is called the Verblunsky parameter and $\rho_j=\sqrt{1-|\gamma_j|^2}$. 
The CMV matrix is expressed by $(SC)^T$. 

\

\noindent{\bf Quantum search driven by coined QW}\\
In particular, if we assign the following local coin operator $\hat{E}_u$ to each $u\in V$ with the marked vertex set $M\subset V$, then 
it is called Szegedy's coined walk with the marked vertices $M$:
Let $\alpha_u$ be a unit vector on $\mathcal{C}_u$, and $f_0:V\to \{0,1\}$ be $f_0(u)=1$ if $u\in M$ and $f_0(u)=0$ if $u\notin M$. 
\begin{enumerate}
\item{Case (I):}
	\begin{align*}
        \hat{E}_u &= 2|\alpha_{u,M}\rangle \langle \alpha_{u,M}|-\bs{1}_{\mathcal{C}_u}, 
        \end{align*}
where $\alpha_{u,M}$ is 
	\begin{align*} 
        \alpha_{u,M}=f_0(u)\alpha_u.  
        \end{align*}
Let $\mathcal{A}_M\subset \ell^2(A)$ be the subspace spanned by the target space as follows:
	\[ \mathcal{A}_M=\mathrm{span}\{ \alpha_u \;|\; u\in M \}. \]
The evolution operator $\hat{U}_M=\hat{S}\hat{C}_M: \ell^2(A)\to \ell^2(A)$ with $\hat{C}_M=\oplus_{u\in V}\hat{E}_u$ 
and the flip-flop shift operator $\hat{S}$ is reexpressed by 
	\[ \hat{U}_M = \hat{U}_{\emptyset}(\Pi_{\mathcal{A}_M^\perp}-\Pi_{\mathcal{A}_M}), \]
where $\hat{U}_{\emptyset}$ is the unitary operator replacing $\alpha_{u,M}$ with a unit vector $\alpha_u(\neq 0)$ for every $u\in V$,
since $\hat{C}_M=\hat{C}_\emptyset-(\hat{C}_{\emptyset}-\hat{C}_M)=\hat{C}_{\emptyset}-2\Pi_{\mathcal{A}_M}=\hat{C_\phi}(1-2\Pi_{\mathcal{A}_M})
=\hat{C_\phi}(\Pi_{\mathcal{A}_M^\perp}-\Pi_{\mathcal{A}_M})$.
\item{Case (II):}
Another natural way of extending is as follows. Let $\gamma_u$ be a unit vector on $\mathcal{C}_u$. Then we define
	\begin{align*}
        \hat{E}_u = (-1)^{f_0(u)}\left( 2|\gamma_u\rangle \langle \gamma_u|-\bs{1}_{\mathcal{C}_u} \right). 
        \end{align*}
\end{enumerate}

\

\noindent{\bf Remark on a vertex based formulation} \\
There is a vertex-based formulation when the coined walk on multigraphs is based on arcs. 
The vertex-based formulation of the coined walk is quite useful when we consider the quantum walk 
on a $d$-dimensional torus lattice $\mathbb{T}_d$ or an infinite lattice $\mathbb{Z}_d$. This formulation is rather familiar for some researchers in the area of quantum walks. 
However, the efficiency of this formulation seems to be restricted to at most a regular multigraph. 
Here we consider only a regular lattice as the graph $G=(V,E)$ for a simplicity. 
Let us consider the Hilbert space 
	\[ \mathcal{H}'=\ell^2(V;\mathbb{C}^{2d})=\{ \psi:V\to \mathbb{C}^{2d} \;|\; \sum_{x\in \mathbb{Z}^d}||\psi(x)||^2_{\mathbb{C}^{2d}}<\infty \}.  \]
The inner product of $\ell^2(\mathbb{Z}^d;\mathbb{C}^{2d})$ is 
	\[ \langle \psi,\varphi \rangle_{\mathcal{H}'}=\sum_{x\in V} \langle \psi(x),\varphi(x) \rangle_{\mathbb{C}^{2d}}. \] 
The dimension of the internal space $\mathbb{C}^{2d}$ corresponds to the direction $\bs{e}_1,-\bs{e}_1,\dots,\bs{e}_{d},-\bs{e}_d$, 
where $\bs{e}_1=[1,0,\dots,0]^T,\;\bs{e}_2=[0,1,\dots,0]^T,\dots,\bs{e}_d=[0,0,\dots,1]^T\in \mathbb{Z}^d$. 
We set the complete orthogonal system of $\mathbb{C}^{2d}$ by $\{|j\rangle,|-j\rangle \;|\; j=1,\dots, d \}$ with 
	\begin{align*} 
	|-1\rangle &= [1,0,\dots,0,0]^T,\; |1\rangle=[0,1,\dots,0,0]^T, \cdots \\
	\cdots, |-d\rangle &= [0,0,\dots,1,0]^T, \; |d\rangle=[0,0,\dots,0,1]^T\in \mathbb{C}^{2d}.
	\end{align*}
The evolution operator $\hat{U}_V=\hat{S}_V\hat{C}_V$ is expressed as 
	\[ (\hat{S}_V\psi)(x)=[\psi_{-1}(x+\bs{e}_1),\psi_{1}(x-\bs{e}_1),\dots,\psi_{-d}(x+\bs{e}_d),\psi_d(x-\bs{e}_d)]^T, \]
for $\psi(x)=[\psi_{-1}(x),\psi_{1}(x),\dots,\psi_{-d}(x),\psi_d(x)]^T$. 

Let $\hat{C}':V\to \{\; 2d$-dimensional unitary operators $\}$. Then 
	\[ (\hat{C}_V\psi)(x)=\hat{C}'(x)\psi(x).  \]
We have the following expression which is a derivation that shows why quantum walks are called quantum analogue of random walks:
	\[ (\hat{U}_V\psi)(x)=\sum_{j=1}^d \hat{P}_{j}(x-\bs{e}_j)\psi(x-\bs{e}_j)+\hat{P}_{-j}(x+\bs{e}_j)\psi(x+\bs{e}_j), \]
where $\hat{P}_{\pm j}(x)=|\pm j\rangle\langle \pm j|\hat{C}'(x)$. 

Each arc $a$ with $t(a)=x$ of $G$ is labeled by 
	\[ a=\begin{cases} (x;-j), & \text{if $o(a)=x+\bs{e}_j$,} \\ (x;j), & \text{if $o(a)=x-\bs{e}_j$,} \end{cases} \]
for $j=1,\dots,d$. 
We define the unitary map from the vertex based space $\ell^2(V;\mathbb{C}^{2d})$ to the arc based space $\ell^2(A)$ as follows: 
	\[  (\Gamma \phi)(x;j)=\langle j,\phi(x)\rangle_{\mathbb{C}^{2d}} \]
for $j\in\{\pm 1,\dots,\pm d\}$. 
The inverse map $\Gamma^{-1} : \ell^2(A)\to \ell^2(V;\mathbb{C}^{2d})$ is 
        \[ (\Gamma^{-1}\psi)(x)=[ \psi(x;-1),\psi(x;1),\dots,\psi(x;-d),\psi(x;d) ]^T. \]
\begin{proposition}
For any vertex-based formulation $\hat{U}_V$, there exists an arc-based formulation $\hat{U}_A$ such that 
	\[ \hat{U}_V=\Gamma^{-1} \hat{U}_A \Gamma. \]
\end{proposition}	
\begin{proof}
Let $\hat{S}_\sigma: \ell^2(V;\mathbb{C}^{2d})\to \ell^2(V;\mathbb{C}^{2d})$ be a permutation operator such that
	\[ (\hat{S}_\sigma\varphi)(x)=[\varphi_1(x+\bs{e}_1),\varphi_{-1}(x-\bs{e}_{1}),\dots,\varphi_d(x+\bs{e}_d),\varphi_{-d}(x-\bs{e}_{d})]^T \]
We have 
	\[ (\Gamma \hat{S}_\sigma \Gamma^{-1}\psi)(x;j)=\psi(x-\mathrm{sgn}(j)\bs{e}_j;-j), \]
which implies $\hat{S}_A:=\Gamma \hat{S}_\sigma \Gamma^{-1}$ is the flip-flop shift operator of $\ell^2(A)$, that is, $(\hat{S}_A\psi)(a)=\psi(\bar{a})$. 
Notice that 
	\[ (\hat{S}_\sigma \hat{S}_V\varphi)(x)=\left(\bigoplus_{j=1}^d \begin{bmatrix} 0 & 1 \\ 1 & 0\end{bmatrix}\right) \varphi(x). \]
Combining the above expressions, we have 
	\begin{align} 
        \Gamma \hat{U}_V \Gamma^{-1} 
        	&= (\Gamma \hat{S}_\sigma \Gamma^{-1}) \cdot (\Gamma \hat{S}_\sigma \hat{S}_V \hat{C} \Gamma^{-1}) \\
        	&= \hat{S}_A \hat{C}_A,
        \end{align}
where $\hat{C}_A=\oplus_{x\in V}\hat{C}'(x):\ell^2(A)\to \ell^2(A)$ such that 
	\begin{equation}\label{swap} 
        \hat{C}'(x)=\Gamma \left(\left(\bigoplus_{j=1}^d \begin{bmatrix} 0 & 1 \\ 1 & 0\end{bmatrix}\right)\hat{C}'(x)\right) \Gamma^{-1}. 
        \end{equation}
\end{proof}
Notice that as expressed by (\ref{swap}), 
the rows which indicate the positive and negative direction of $\bs{e}_j$ are swapped between the local coin operators 
in the vertex and arc representations $(j=1,\dots,d)$. The shift operator $\hat{S}_\sigma$ is called the flip-flop shift and $\hat{S}_V$ is called the moving shift. 

\subsection{Staggered walk}
A quantum walk is a staggered walk on a connected simple graph $G=(V,E)$ when it is based on a tessellation cover. A tessellation ${\mathcal{T}}$ is a partition of the graph into cliques\footnote{A clique of $G$ is a set of vertices that induces a complete subgraph of $G$.}, where each partition element is called a polygon. A tessellation cover is a set of tessellations $\{{\mathcal{T}}_1,...,{\mathcal{T}}_k\}$ that covers the graph edges, that is, $\cup_{\ell=1}^k {\mathcal{E}}({\mathcal{T}}_\ell)=E$, where ${\mathcal{E}}({\mathcal{T}})$ is the set of edges of tessellation ${\mathcal{T}}$. An edge belongs to a tessellation if the vertices incident to the edge belongs to the same polygon. When the tessellation cover has size $k$, the graph is called $k$-tessellable~\cite{Por1,Por2}. Given a graph $G$, an interesting problem in graph theory is to determine the minimum size of a tessellation cover of $G$.

In this work we address only 2-tessellable staggered walks. A graph $G$ is 2-tessellable if and only if the clique graph $K(G)$ is 2-colorable~\cite{Por2}. It is known that a graph $G$ has a 2-colorable clique graph if and only if $G$ is the line graph of a bipartite multigraph~\cite{Peterson}. Then, in our case, $G$ is the line graph of a bipartite multigraph. Notice that the line graph of a bipartite multigraph is a simple graph.

Suppose that graph $G=(V,E)$ admits a tessellation cover $\{{\mathcal{T}}_1,{\mathcal{T}}_2\}$ of size 2. Let ${\mathcal{T}}_1=\{K_1,...,K_{|{\mathcal{T}}_1|}\}$, where each $K_p$ is a polygon of ${\mathcal{T}}_1$ and $|{\mathcal{T}}_1|$ is the number of polygons in ${\mathcal{T}}_1$.  Besides, $K_p$ is a clique and $K_p\cap K_{p'}=\emptyset$ for $p\not=p'$ and $\cup_{p=1}^{|{\mathcal{T}}_1|} K_p=V$. Likewise, ${\mathcal{T}}_2=\{K_1' ,...,K_{|{\mathcal{T}}_2|}'\}$, where the set $\{K_1' ,...,K_{|{\mathcal{T}}_2|}'\}$ is a second partition of the graph into cliques. The tessellation union must cover the graph edges, that is, ${\mathcal{E}}({\mathcal{T}}_1)\cup {\mathcal{E}}({\mathcal{T}}_2)=E$, where ${\mathcal{E}}({\mathcal{T}}_1)=\cup_{p=1}^{|{\mathcal{T}}_1|} E(K_p)$ and ${\mathcal{E}}({\mathcal{T}}_2)=\cup_{p=1}^{|{\mathcal{T}}_2|} E(K_p')$.

In the framework of two-partition walks, we set $\Omega=V$ and define $\pi_1$ and $\pi_2$ as 
	\begin{align*}
        u \overset{\pi_1}{\sim} v & \,\,{\Longleftrightarrow}\,\, \exists p\in {\mathcal{T}}_1 \mathrm{\;such\; that\;} u,v\in V(K_p),  \\
        u \overset{\pi_2}{\sim} v & \,\,{\Longleftrightarrow}\,\, \exists q\in {\mathcal{T}}_2 \mathrm{\;such\; that\;} u,v\in V(K_q'). 
        \end{align*} 
The respective quotient sets of $V$ by $\pi_1$ and $\pi_2$ are
	\begin{align*}
        V/\pi_1 &= \{ C_p \;|\; p\in {\mathcal{T}}_1  \} \cong {\mathcal{T}}_1, \\
        V/\pi_2 &= \{ D_q \;|\; q\in {\mathcal{T}}_2  \}\cong {\mathcal{T}}_2,
        \end{align*}
where $C_p=\{u\in V \;|\; u\in V(K_p) \}$ and $D_q=\{u\in V \;|\; u\in V(K_q) \}$. 
Thus, the staggered walk is determined by $(G;{\mathcal{T}}_1,{\mathcal{T}}_2;\hat{U})$, 
where $G$ is a 2-tessellable simple graph; 
${\mathcal{T}}_1$ and ${\mathcal{T}}_2$ are tessellations of $G$; and the evolution operator is $\hat{U}=\hat{F}\,\hat{E}$, where $\hat{E}=\oplus_{p\in |{\mathcal{T}}_1|}\hat{E}_p$ and
$\hat{F}=\oplus_{q\in |{\mathcal{T}}_2|} \hat{F}_q$. 

The association between a tessellation and a unitary operator in the staggered model is performed in the way described in~\cite{Por1,Por3}. Here we extend this connection. Consider tessellation ${\mathcal{T}}_1$, which is the set of polygons $K_p$ for $1\le p \le |{\mathcal{T}}_1|$ and tessellation ${\mathcal{T}}_2$, which is the set of polygons $K_q'$ for $1\le q \le |{\mathcal{T}}_2|$. A polygon $K_p$ must be associated with a Hermitian operator $\hat{H}_p$ in the Hilbert space ${\mathcal{C}_p}$ spanned by the vertices of $K_p$ and, likewise, a polygon $K_q'$ must be associated with a Hermitian operator $\hat{H}_p'$ in the Hilbert space ${\mathcal{D}_q}$ spanned by the vertices of $K_q'$. Any choice of $\hat{H}_p$ and $\hat{H}_p'$ is acceptable as long as $\hat{H}_p$ and $\hat{H}_p'$ are Hermitian. A natural way to choose $\hat{H}_p$ and $\hat{H}_p'$ is to use a classical Markov chain with symmetric transition matrix or to use the adjacency matrix $A$ of $G$. $\hat{H}_p$ and $\hat{H}_q'$ are obtained from $A$ by deleting the lines and columns of $A$ associated with the vertices $V\setminus K_p$ and $V\setminus K_q'$, respectively. Notice that the Markov chain and the staggered walk are defined on the same graph $G$. Following \cite{Por3}, the local unitary operators $\{\hat{E}_p\}_{p\in {\mathcal{T}}_1}$ and $\{\hat{F}_q\}_{q\in {\mathcal{T}}_2}$ 
are defined as
	\begin{align*}
        \hat{E}_p &= \exp(i\theta_1 \hat{H}_p), \\
        \hat{F}_q &= \exp(i\theta_2 \hat{H}_q'),
        \end{align*}
where $\theta_1$ and $\theta_2$ are angles. 

An interesting form for operators $\hat{H}_p$ and $\hat{H}_q'$ discussed in \cite{Por3} is
	\begin{align*}
        \hat{H}_p &= 2|\alpha_p\rangle \langle \alpha_p|-\bs{1}_{\mathcal{C}_p}, \\
        \hat{H}_q' &= 2|\beta_q\rangle \langle \beta_q|-\bs{1}_{\mathcal{D}_q},
        \end{align*}
where $|\alpha_p\rangle$ and $|\beta_q\rangle$ are unit vectors in ${\mathcal{C}_p}$ and ${\mathcal{D}_q}$, respectively. Notice that in this case
	\begin{align*}
        \hat{E}_p &= \cos(\theta_1)\bs{1}_{\mathcal{C}_p}+i\sin(\theta_1)\hat{H}_p, \\
       \hat{F}_q &= \cos(\theta_2)\bs{1}_{\mathcal{D}_q}+i\sin(\theta_2)\hat{H}_q'.
        \end{align*}

\noindent{\bf Quantum search in the staggered model} \\
One of the most interesting method to search a marked vertex assuming that the graph $G$ has $M$ marked vertices is use the query operator
\begin{equation}
\hat{U}_M=2\sum_{u\in M} |u\rangle\langle u| - \bs{1}_{\mathcal{H}}.
\end{equation}
In this case, the evolution operator is $\hat{U}=\hat{F}\,\hat{E}\, \hat{U}_M$. This method is similar to the one used in the Grover algorithm. There is a slight variation that uses the operator  $\hat{U}=\hat{F}\, \hat{U}_M\,\hat{E}\, \hat{U}_M$. Ref.~\cite{APN} used the query-based method to show an example which is quadratically faster compared to random-walk based algorithms on the same graph.

The query-based search does not directly reproduce the searching method employed in the extended Szegedy model. To exactly reproduce Szegedy's method, it is necessary to introduce the concept of partial tessellations, which was addressed in~\cite{Por1,Por4}. In the staggered model, it is not necessary to modify the graph in order to search for a marked vertex. The concept of partial tessellation exactly reproduce the method that uses sinks in directed multigraphs in Szegedy's model.

\subsection{Quantum walk on hypergraphs}

We propose a quantum walk on hypergraph $H=(V, \mathcal{E})$,
where $V$ is a discrete set called the vertex set and $\mathcal{E}\subseteq 2^{V}$ is called the hyperedge set. 
If $v,u\in e \in \mathcal{E}$, then we say that $u$ and $v$ are adjacent. 
In particular, if we set $\mathcal{E}\subseteq \binom{V}{2}$, then the hypergraph reduces to a graph.
In the framework of two-partition walks, we set 
	\[ \Omega = \mathcal{A}:=\{(e,u) \;|\; e\in \mathcal{E},\; u\in e\},  \]
and define the equivalence relations $\pi_1$ and $\pi_2$ as
	\begin{align*}
        (e,u) \overset{\pi_1}{\sim} (e',u') & \,\,{\Longleftrightarrow}\,\, u=u',  \\
        (e,u) \overset{\pi_2}{\sim} (e',u') & \,\,{\Longleftrightarrow}\,\, e=e'. 
        \end{align*}
The quotient sets are 
	\begin{align*}
        \mathcal{A}/\pi_1 &= \{ \{(e,u) \;|\; \forall e,\; u\in e\} \;|\; u\in V\}, \\
        \mathcal{A}/\pi_2 &= \{ \{(e,u) \;|\; \forall u,\; u\in e\} \;|\; e\in \mathcal{E}\}.
        \end{align*}
When we take $\mathcal{E}\subseteq \binom{V}{2}$, $\mathcal{A}$ is isomorphic to the symmetric arc set $A$ of the graph induced
by the following bijection $\phi: \mathcal{A}\to A$. If $e=\{u,v\}\in \mathcal{E}$, then 
	\[ t(\phi((e,u))) = u,\; o(\phi((e,u)))=v. \]
Thus the inverse is expressed by 
	\[ \phi^{-1}(a)=(|a|,t(a)),\;a\in A. \]

\begin{remark}
Assume that $\mathcal{E}\subseteq \binom{V}{2}$. 
Using the above bijection map, we define $a=\phi(e,u)$, $b=\phi(e',u')$. 
Then, we have 
	\begin{align*}
        (e,u) \overset{\pi_1}{\sim} (e',u') &\,\,\Longleftrightarrow\,\, t(a)=t(b), \\
        (e,u) \overset{\pi_2}{\sim} (e',u') &\,\,\Longleftrightarrow\,\, |a|=|b|.
        \end{align*}
Thus, $a$ and $b$ satisfy the equivalence relations $\pi'_1$ and $\pi'_2$ of the arc set of the graph for the coined walk case in Sec.~3.2. 
This quantum walk is naturally extended from the coined walk on a simple graph to a quantum walk on a hypergraph. 
\end{remark}
\begin{table}
\begin{center}
\begin{tabular}{|c|c|c|c|c|} \hline
Walk & {\bf Bipartite }($\mathcal{B}$) & {\bf Coined}($\mathcal{C}$) & {\bf Staggered}($\mathcal{S}$) & {\bf Hypergraph} \\ 
 & {\small (Sec.~3.1)}                & {\small (Sec.~3.2)}             & {\small (Sec.~3.3)}                & {\small (Sec.~3.4)} \\ \hline
         & edge set          & symmetric arc  & vertex set              & pair of hyperedges  \\ 
$\Omega$ &  of   a  bipartite           & set of  a               & of  a                 & and  its\\ 
         & multigraph & multigraph            & 2-tessellable graph & contained vertex \\ \hline
$\Omega/\pi_1$
         & $X$-end vertices    & terminal vertices  & tessellation ${\mathcal{T}}_1$  & vertices \\ \hline
$\Omega/\pi_2$
         & $Y$-end vertices    & edges             & tessellation ${\mathcal{T}}_2$  & edges \\ \hline
\end{tabular}
\caption{Examples of two-partition quantum walks.}\label{table1}
\end{center}
\end{table}

\section{Unitary equivalence of quantum walks}

Let $\mathcal{P}$, $\mathcal{B}$, and $\mathcal{S}$ be the family of all two-partition walks,  all bipartite walks, and all 2-tessellable staggered walks, respectively. 
The family of all coined walks is denoted by $\mathcal{C}$, which has evolution operator $\hat{\Gamma}=\hat{\Gamma}(G;\{C_u\}_{u\in V})$. 
We also define the family of the two-step coined quantum walks, which is denoted by $\mathcal{C}_2$ and has $\hat{\Gamma}^2$ as evolution operator. 
The two-step coined walk can also be formulated in terms of the two-partition walk model, see Lemma~\ref{lem} for details. 
Table~\ref{table1} summarizes the quantum walk families analyzed in this work.
Each quantum walk model is described by $(L;\hat{\Theta})$, where $L$ is the discrete set $K$ (walker's positions) together with two partitions, and $\hat{\Theta}$ is the evolution operator. 
The evolution operator $\hat{\Theta}$ acts on $\ell^2(K)$. 
Table~\ref{table2} describes $L$, $K$, and $\hat{\Theta}$ for each family of quantum walk model. 
We define an order between the families of quantum walk models as follows. 
\begin{definition}
Assume that $\mathcal{A}\in \{\mathcal{P},\mathcal{B},\mathcal{C},\mathcal{C}_2,\mathcal{S}\}$.
For any quantum walk in $\mathcal{A}$ with the evolution operator $\hat{\Theta}$ that acts on $\ell^2(K)$, 
if there exists a quantum walk in $\mathcal{A}'\in \{\mathcal{P},\mathcal{B},\mathcal{C},\mathcal{C}_2,\mathcal{S}\}$ with evolution operator 
$\hat{\Theta}'$ that acts on  $\ell^2(K')$, 
and an injection map $\eta:K\longrightarrow K'$, such that
	\[ \hat{\Theta}=\mathcal{U}^{-1}_\eta \;\hat{\Theta}' \; \mathcal{U}_\eta, \]
then we denote $\mathcal{A} \prec \mathcal{A}'$. 
Here $\mathcal{U}_\eta:\ell^2({K})\to \ell^2(\eta({K}))$ is the unitary map, that is, $(\mathcal{U}_\eta\psi)(a)=\psi(\eta^{-1}(a))$. 
In particular if the converse also holds, that is, $\mathcal{A}\succ \mathcal{A}'$, then we denote $\mathcal{A} \cong \mathcal{A}'$.
\end{definition}

\begin{table}
\begin{center}
\begin{tabular}{|c|c|c|c|c|c|} \hline
    & $\mathcal{P}$        & $\mathcal{B}$ & $\mathcal{C}$ & $\mathcal{C}_2$ & $\mathcal{S}$ \\ \hline
$L$ & $\Omega;\pi_1,\pi_2$ & bipartite    &  multigraph       &   multigraph      & 2-tessellable \\
    &                      & multigraph  &                    &                    & graph; $\mathcal{T}_1,\mathcal{T}_2$  \\ \hline

   &                       &             & induced            & induced             &      \\               
$K$ & $\Omega$             & edge set     &        symmetric  &      symmetric  & vertex set \\               
    &                      &              & arc set            &  arc set           &            \\ \hline 
$\hat{\Theta}$ & $(\oplus_{i\in \Omega/\pi_2}\hat{F}_i)$           & $(\oplus_{y\in Y}\hat{R}_y)$          & $\hat{S}\cdot(\oplus_{u\in V} \hat{C}_u)$  & $(\hat{S}\cdot(\oplus_{u\in V}\hat{C}_u))^2$   & $(\oplus_{q\in {\mathcal{T}}_2}\hat{F}_q)$ \\
               & $\;\;\cdot(\oplus_{j\in \Omega/\pi_1} \hat{E}_j)$ & $\;\;\cdot(\oplus_{x\in X} \hat{R}_x)$ & $\;\;$               & $\;\;$               & $\;\;\cdot(\oplus_{p\in {\mathcal{T}}_1} \hat{E}_p)$ \\ \hline                            
\end{tabular}
\caption{Family of quantum walks $\mathcal{P}$, $\mathcal{B}$, $\mathcal{C}$, $\mathcal{C}_2$, and $\mathcal{S}$, which are characterized by $(L,\hat{\Theta})$, where $L$ is the discrete set $K$ together with two partitions, 
and $\hat{\Theta}$ is the evolution operator.}\label{table2}
\end{center}
\end{table}

The previous examples in Secs.~3.1, 3.2, and 3.3 show that $\mathcal{B},\mathcal{C},\mathcal{S} \prec \mathcal{P}$, respectively. 
Next, we show ``$\mathcal{B}\succ \mathcal{P}$" in Sec.~4.1, ``$\mathcal{S} \succ \mathcal{P}$" in Sec.~4.2, 
``$\mathcal{B} \prec \mathcal{C}_2$" in Sec.~4.3, and ``$\mathcal{B} \succ \mathcal{C}_2$" in Sec.~4.4. 
See also Fig.~1 for the commutative diagram. 
In this section, we show the following theorem:
\begin{theorem}\label{unitaryeq}
Let $\mathcal{P},\mathcal{B},\mathcal{C},\mathcal{C}_2,\mathcal{S}$ be the families above defined. 
Then, 
\[ \mathcal{C}\prec \mathcal{B}\cong \mathcal{P}\cong \mathcal{S} \cong \mathcal{C}_2. \]
See the commutative diagram in Figs.~\ref{commutativemap} and the injection maps in Table.~\ref{table3}.
\end{theorem}
\begin{figure}[htbp]
  \begin{center}
           \includegraphics[clip, width=7cm]{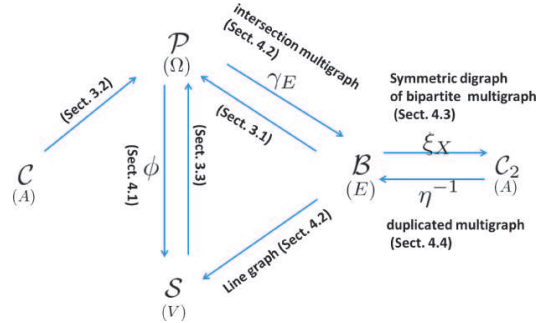}
              \caption{Commutative diagram of $\mathcal{B}$, $\mathcal{P}$, $\mathcal{S}$ and $\mathcal{C}_2$. }
              \label{commutativemap}
  \end{center}
\end{figure}
\begin{table}[htbp]
\begin{center}
\begin{tabular}{|l|} \hline
$\gamma_E:\Omega\to E(\Omega;\pi_1,\pi_2)$ \\
 \quad $X(\gamma_E(\omega))=\gamma_V(C(\omega))$, \; $Y(\gamma_E(\omega))=\gamma_V(D(\omega))$ \\ \hline
$\xi_X:E(G)\to A_X(G)$ \\
 \quad $t(\xi_X(e))=X(e)$, \; $o(\xi_X(e))=Y(e)$ \\ \hline
$\eta^{-1}:A(G)\to E(G_2)$ \\
 \quad $V(\eta^{-1}(a))=t(a)$, \; $V'(\eta^{-1}(a))=o(a)$ \\ \hline 
\end{tabular}
\caption{Maps $\gamma_E$, $\xi_X$, and $\eta^{-1}$ (see also Fig.~2).}
\label{table3}
\end{center}
\end{table}
\begin{figure}[htbp]
  \begin{center}
            \includegraphics[clip, width=7cm]{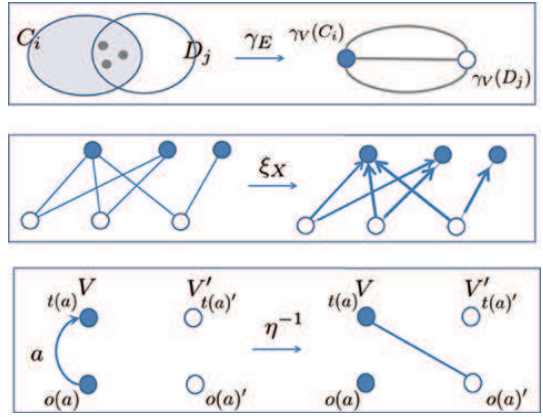}
              \caption{Maps $\gamma_E$, $\xi_X$, and $\eta^{-1}$. }
  \end{center}
\end{figure}

\subsection{Proof of $\mathcal{P} \prec \mathcal{S}$}\label{H}
We define the following simple graph by $H(\Omega;\pi_1,\pi_2)$. 
Here 
\[ V = \phi(\Omega), \]
where $\phi$ is a bijection map from $\Omega$ to $V$. 
For $u,v\in V$ with $u\neq v$, 
\[ \{u,v\}\in E  \Leftrightarrow \phi^{-1}(u)\stackrel{\pi_1}{\sim} \phi^{-1}(v)\mathrm{\;or\;} \phi^{-1}(u)\stackrel{\pi_2}{\sim} \phi^{-1}(v).  \]
It is obvious that this graph is 2-tessellable and tessellation ${\mathcal{T}}_1$ is isomorphic to $\Omega/\pi_1$ and 
tessellation ${\mathcal{T}}_2$ is isomorphic to $\Omega/\pi_2$. 
We have ${\mathcal{T}}_1=\phi(\Omega/\pi_1)$, ${\mathcal{T}}_2=\phi(\Omega/\pi_2)$, 
and $\phi^{-1}(p)=\{\omega \;|\; \phi(\omega)\in p\}\in \Omega/\pi_1$, 
$\phi^{-1}(q)=\{\omega \;|\; \phi(\omega)\in q\}\in \Omega/\pi_1$ for $p\in {\mathcal{T}}_1$, $q\in {\mathcal{T}}_2$. 
Then, we have the following proposition which completes the proof of $\mathcal{P}\prec \mathcal{S}$: 
\begin{proposition}
Given $\hat{U}=\left(\oplus_{j\in [J_2]}\hat{F}_j\right)\left(\oplus_{i\in [J_1]}\hat{E}_i\right) \in \mathcal{P}$ in $(\Omega;\pi_1,\pi_2)$, 
let $H=(V,E)$ be the above 2-tessellable graph. 
Let $\mathcal{U}_\phi: \ell^2(\Omega)\to \ell^2(V)$ 
be a unitary map such that $(\mathcal{U}_\phi \psi)(v)=\psi(\phi^{-1}(v))$. 
Then, there exists $\hat{R}=\left(\oplus_{q\in {\mathcal{T}}_2}\hat{F}'_q\right)\left(\oplus_{p\in {\mathcal{T}}_1}\hat{E}'_p\right)$ on $\ell^2(V)$ under the clique decompositions 
${\mathcal{T}}_1=\phi(\Omega/\pi_1)$ and ${\mathcal{T}}_2=\phi(\Omega/\pi_2)$, 
such that 
	\[ \hat{U}=\mathcal{U}_\phi^{-1}\; \hat{R} \;\mathcal{U}_\phi, \]
which implies $\mathcal{P}\prec \mathcal{S}$. 
Here, $\hat{F}'_q= \mathcal{U}_\phi \hat{F}_{\phi^{-1}(q)}\mathcal{U}_\phi^{-1}$ and $\hat{E}'_p= \mathcal{U}_\phi\hat{E}_{\phi^{-1}(p)}\mathcal{U}_\phi^{-1}$. 
\end{proposition}
\begin{proof}
We show that $\mathcal{U}_\phi\hat{U}\mathcal{U}_\phi^{-1}$ is an evolution operator of a 2-tessellable staggered walk on $H=(V,E)$ induced by $(\Omega;\pi_1,\pi_2)$. 
The operator $\mathcal{U}_\phi\hat{U}\mathcal{U}_\phi^{-1}$ is a unitary operator on $\ell^2(V)$ since we just take a relabeling the standard bases of 
$\ell^2(\Omega)$ by the bijection map $\phi$. 
Thus, the problem is reduced to show that $\mathcal{U}_\phi\hat{E}_{j}\mathcal{U}_\phi^{-1}$ and $\mathcal{U}_\phi \hat{F}_{i}\mathcal{U}_\phi^{-1}$ are 
local unitary operators on $\spann\{\delta_u \;|\; u\in \phi(C_j)\}$, $\spann\{\delta_u \;|\; u\in \phi(D_i)\}\subset \ell^2(V)$, respectively, 
since $\mathcal{T}_i = \phi(\Omega/\pi_i)$ $(i=1,2)$.  
It is sufficient to show the locality because it is clear that they are unitary. 

We put $u=\phi(\omega)$ and $v=\phi(\omega')$. 
Notice that 
	\[ u\in \phi(C_j) \Leftrightarrow \phi(\omega)\in \phi(C_j) \Leftrightarrow \omega\in C_j.   \]
Using this, we have  
	\begin{align*}
        u\notin \phi(C_j) \;\mathrm{or}\; v\notin \phi(C_j)
        	& \Leftrightarrow \omega\notin C_j \;\mathrm{or}\; \omega'\notin C_j \\
                & \Rightarrow \langle \delta_\omega, \hat{E}_j\delta_{\omega'}  \rangle=0
        \end{align*}
since $\hat{E}_j$ is a local operator on $\spann\{\delta_\omega \;|\; \omega\in C_j\}\subset \ell^2(\Omega)$. 
Thus 
	\begin{align*}
        u\notin \phi(C_j) \;\mathrm{or}\; v\notin \phi(C_j)
        	& \Rightarrow \langle \mathcal{U}^{-1}_\phi\delta_{\phi(\omega)}, \hat{E}_j\mathcal{U}^{-1}_\phi\delta_{\phi(\omega')}  \rangle=0 \\
                & \Leftrightarrow \langle \delta_{\phi(\omega)}, \mathcal{U}_\phi\hat{E}_j\mathcal{U}^{-1}_\phi\delta_{\phi(\omega')}  \rangle=0
        \end{align*}
which implies that 
$\mathcal{U}_\phi\hat{E}_{j}\mathcal{U}_\phi^{-1}$ is a local operator on $\spann\{\delta_u \;|\; u\in \phi(C_j)\}$.
\end{proof}
\begin{figure}[htbp]
  \begin{center}
            \includegraphics[clip, width=7cm]{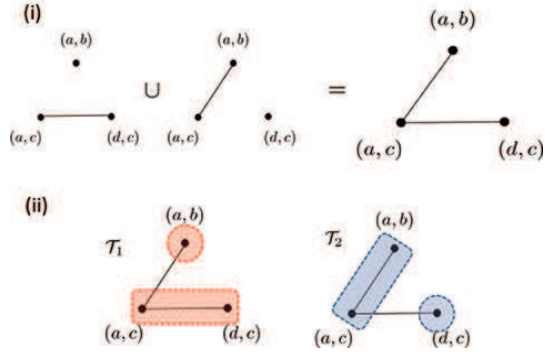}
              \caption{The graph induced by $(\Omega;\pi_1,\pi_2)$ of Example~\ref{ex1}: Figure (i) depicts the subgraphs induced by $\pi_1$ and $\pi_2$, respectively, whose union is the graph associated with Example~\ref{ex1}. 
              Fig.~(ii) depicts tessellations $\mathcal{T}_1$ and $\mathcal{T}_2$. }
  \end{center}
\end{figure}
\subsection{Proof of $\mathcal{P}\prec \mathcal{B}$}
Given a two-partition walk $(\Omega;\pi_1,\pi_2)$ with $\Omega/\pi_1=\{C_i\}$, $\Omega/\pi_2=\{D_j\}$, 
we define $C(\omega)=C_i$ and $D(\omega)=D_j$ for any $\omega\in C_i\cap D_j$. 
\begin{definition}\label{IntersectionGraph}
Let $\Omega$ be a discrete set and $\pi_1,\pi_2$ be partitions, that is, 
$\Omega/\pi_1=\{C_i\}_{i\in [J_1]}$, $\Omega/\pi_2=\{D_j\}_{j\in [J_2]}$. 
The generalized intersection graph induced by $(\Omega;\pi_1,\pi_2)$, 
$G(\Omega;\pi_1,\pi_2)=(X\sqcup Y, E)$, is defined as follows: 
	\[ X=\gamma_V(\Omega/\pi_1),\;Y=\gamma_V(\Omega/\pi_2),  \]
	\[ E=\gamma_E(\Omega). \]
The bijection maps $\gamma_V: \Omega/\pi_1\cup \Omega/\pi_2 \to X\cup Y$ and $\gamma_E: \Omega\to E$ are defined as follows: 
	\[ \gamma_V(C_i)=i,\; \gamma_V(D_j)=j, \] 
	\[X(\gamma_E(\omega))=\gamma_V(C(\omega)),\; Y(\gamma_E(\omega))=\gamma_V(D(\omega)).\] 
\end{definition}
\begin{figure}[htbp]
  \begin{center}
            \includegraphics[clip, width=5cm]{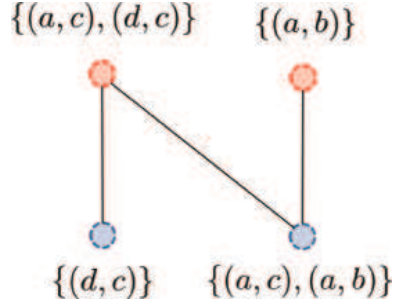}
              \caption{The intersection graph induced by $(\Omega;\pi_1,\pi_2)$ in Example~\ref{ex1}}. 
  \end{center}
\end{figure}
The graph $G(\Omega;\pi_1,\pi_2)$ is a bipartite multigraph; 
the multiplicity between $x\in X$ and $y\in Y$ is described by $|\gamma_V^{-1}(x)\cap \gamma_V^{-1}(y)|$. 
Conversely, given an arbitrary connected bipartite multigraph $G=(X\sqcup Y, E)$, we can induce $(\Omega;\pi_1,\pi_2)$ 
as follows: 
$\Omega=E$, $e\stackrel{\pi_1}{\sim} f \Leftrightarrow X(e)=X(f)$, $e\stackrel{\pi_2}{\sim} f \Leftrightarrow Y(e)=Y(f)$. 
Therefore, the set of all connected bipartite multigraphs and the set of all $(\Omega; \pi_1,\pi_2)$ are isomorphic. 
Using the above bijection map $\gamma_E: \Omega\to E$, we have the following proposition which completes the proof of $\mathcal{P}\prec \mathcal{B}$: 
\begin{proposition}
Let $\mathcal{U}_{\gamma_E}:\ell^2(\Omega)\to \ell^2(E)$ be
	\[ (\mathcal{U}_{\gamma_E}\psi)(e)=\psi(\gamma_E^{-1}(e)), \]
where $E$ is the edge set induced by $(\Omega;\pi_1,\pi_2)$.
Then, for any two-partition walk $ \hat{U} := \hat{U}( \Omega; \pi_1, \pi_2; \allowbreak \{\hat{E}_i\}, \{\hat{F}_j\}) \in \mathcal{P} $, 
there exists a bipartite walk $\hat{W}\in \mathcal{B}$ on the generalized intersection graph $(X\sqcup Y, E)$ of $(\Omega;\pi_1,\pi_2)$ 
with $\{\hat{R}_i\}_{i\in X}$, $\{\hat{R}_j\}_{j\in Y}$ such that
	\[ \hat{U}=\mathcal{U}_{\gamma_E}^{-1}\hat{W}\mathcal{U}_{\gamma_E}. \]
Here $X=\gamma_V(\Omega/\pi_1),\;Y=\gamma_V(\Omega/\pi_2)$ and $\hat{R}_i=\mathcal{U}_{\gamma_E} \hat{E}_i \mathcal{U}_{\gamma_E}^{-1}$,
\;$\hat{R}_j=\mathcal{U}_{\gamma_E} \hat{F}_j \mathcal{U}_{\gamma_E}^{-1}$. 
\end{proposition}
\begin{proof}
It is sufficient to show that $\mathcal{U}_{\gamma_E} \hat{E}_i \mathcal{U}_{\gamma_E}^{-1}$ and 
$\mathcal{U}_{\gamma_E} \hat{F}_j \mathcal{U}_{\gamma_E}^{-1}$ are local operators on $\spann\{\delta_e \;|\; X(e)=\gamma_V(C_i)\}$ and 
$\spann\{\delta_e \;|\; Y(e)=\gamma_V(D_j)\}\subset \ell^2(E)$. respectively. 
Putting $e=\gamma_E(\omega)$, $f=\gamma_E(\omega')$, we have the following equivalent deformation as follows:
\begin{align*}
X(e)\neq \gamma_V(C_i) \;\mathrm{or}\; X(f)\neq \gamma_V(C_i) 
	 & \Leftrightarrow \gamma_V(C(\omega)) \neq \gamma_V(C_i) \;\mathrm{or}\; \gamma_V(C(\omega')) \neq \gamma_V(C_i) \\
         & \Leftrightarrow C(\omega) \neq C_i \;\mathrm{or}\; C(\omega') \neq C_i \\
         & \Leftrightarrow \omega\notin C_i \;\mathrm{or}\; \omega'\notin C_i
\end{align*}
Since $\hat{E}_i$ is a local operator on $\spann\{\delta_\omega \;|\; \omega\in C_i\}$, then 
\begin{align*}
X(e)\neq \gamma_V(C_i) \;\mathrm{or}\; X(f)\neq \gamma_V(C_i) 
	& \Rightarrow \langle \delta_\omega,\hat{E}_i\delta_{\omega'} \rangle=0 \\
        & \Leftrightarrow \langle \mathcal{U}^{-1}_{\gamma_E}\delta_{\gamma_E(\omega)},\hat{E}_i\mathcal{U}^{-1}_{\gamma_E}\delta_{\gamma_E(\omega')} \rangle=0 \\
        & \Leftrightarrow \langle \delta_{e},\mathcal{U}_{\gamma_E}\hat{E}_i\mathcal{U}^{-1}_{\gamma_E}\delta_{f} \rangle=0.
\end{align*}
Therefore $\mathcal{U}_{\gamma_E} \hat{E}_i \mathcal{U}_{\gamma_E}^{-1}$ is a local operator on $\spann\{\delta_e \;|\; X(e)=\gamma_V(C_i)\}$. 
In the same way, we can show that $\mathcal{U}_{\gamma_E} \hat{F}_j \mathcal{U}_{\gamma_E}^{-1}$ is a local operator on $\spann\{\delta_e \;|\; Y(e)=\gamma_V(D_j)\}$.
\end{proof}

From Secs. 4.1 and~4.2, we obtain automatically the equivalence relation between $\mathcal{B}$ and $\mathcal{S}$. 
The line graph of $G=(V,E)$, $L(G)=(V_L,E_L)$, is defined as follows: 
	\[ V_L=E; \]
        \[ E_L=\{\{e,f\} \;|\; e\cap f\in V \}. \]
\begin{remark}
The graph $H(\Omega;\pi_1,\pi_2)$ in Sec.~\ref{H} is the line graph of $G(\Omega;\pi_1,\pi_2)$ of Definition~\ref{IntersectionGraph}. 
\end{remark}
\subsection{Proof of $\mathcal{B} \cong \mathcal{C}_2$}\label{GG}
As a preparation for the proof, we reexpress $\mathcal{C}_2$, whose evolution operator is described by two steps of a coined walk, 
in the framework of the two-partition quantum walks, which will be useful for the proof. 
\begin{lemma}\label{lem}
Every two-step coined walk on multigraph $G=(V,E)$ is formulated by a two-partition walk 
$(\Omega;\pi_1,\pi_2;\{\hat{E}_u\}_{u\in V},\{\hat{F}_v\}_{v\in V})$, where $\Omega= A$,
\begin{align*}
a\stackrel{\pi_1}{\sim}b &\Leftrightarrow t(a)=t(b),\\
a\stackrel{\pi_2}{\sim} b &\Leftrightarrow o(a)=o(b), 
\end{align*}
and $\hat{F}_u =\hat{S}\hat{E}_u\hat{S}$, $u\in V$. 
\end{lemma}
\begin{proof}
The evolution operator of the two-step coined walk is described by 
	\[ \hat{\Gamma}_2=(\hat{S}\hat{C}\hat{S})\, \hat{C}, \]
where $\hat{S}$ and $\hat{C}=\oplus_{u\in V}\hat{E}_u$ are the shift and coin operators, respectively. 
The coin operator is the direct sum of local unitary operators $\hat{E}_u$. 
Since $\hat{E}_u$ is a local operator on $\mathcal{C}_u$, it holds
	\begin{align*}
        t(a) \neq u\mathrm{\;or\;} t(b) \neq u \Rightarrow \langle \delta_b,\hat{E}_u\delta_a \rangle=0. 
        \end{align*}
This is equivalent to
	\begin{align*}
        o(a) \neq u\mathrm{\;or\;} o(b) \neq u \Rightarrow \langle \delta_{b},\hat{S}\hat{E}_u\hat{S}\delta_{a} \rangle=0,
        \end{align*}
since $\hat{S}$ flips the direction of each arc. 
Therefore $\hat{S}\hat{C}\hat{S}$ follows the decomposition $A/\pi_2$ and 
the local unitary operators $\{\hat{F}_u\}_{u\in V}$ are $\{S\hat{E}_uS\}_{u\in V}$. 
\end{proof}
\subsubsection{Proof of $\mathcal{B}\prec \mathcal{C}_2$}
For given $\hat{W}\in \mathcal{B}$ with $G=(X\sqcup Y,E)$ and $\{\hat{R}_x\}_{x\in X}$, $\{\hat{R}_y\}_{y\in Y}$, 
we will show that $\hat{W}$ is expressed by some $(\oplus_{u\in V}\hat{F}_u)(\oplus_{u\in V}\hat{E}_u)\in \mathcal{C}^2$ using Lemma~\ref{lem}. 

Let $A$ be the set of symmetric arcs induced by $E$ for given bipartite multigraph $G=(X\sqcup Y,E)$. 
We define injection maps $\xi_X,\xi_Y: E\to A$ such that 
	\begin{align*} 
        t(\xi_X(e))\in X, &\;o(\xi_X(e))\in Y, \\
        t(\xi_Y(e))\in Y, &\;o(\xi_X(e))\in X.
        \end{align*}
Setting $A_X=\{ a\in A \;|\; t(a)\in X \}$ and $A_Y=\{ a\in A \;|\; t(a)\in Y \}$, we have $\xi_X(E)=A_X$, $\xi_Y(E)=A_Y\subset A$. 
The inverse maps restricted to the domains by $A_X$ and $A_Y$ are $\xi_X^{-1}(a)=|a|$, $\xi_Y^{-1}(a)=|a|$, respectively. 
We define a unitary map $\mathcal{U}_{\xi_Z}: \ell^2(E)\to \ell^2(A_Z)$ by 
	\[ (\mathcal{U}_{\xi_Z}\psi)(a)=\psi(\xi_Z^{-1}(a)) \;(Z=X,Y).\]
\begin{figure}[htbp]
  \begin{center}
            \includegraphics[clip, width=5cm]{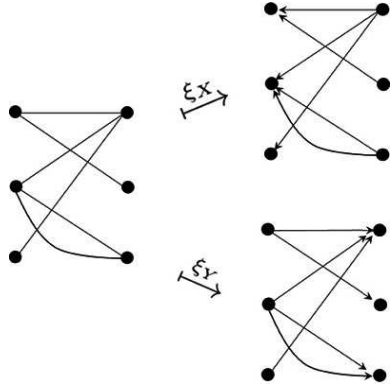}
              \caption{The maps $\xi_X$ and $\xi_Y$.} 
  \end{center}
\end{figure}
Using these unitary maps, we obtain the following proposition which implies $\mathcal{B} \prec \mathcal{C}_2$. 
\begin{proposition}
For any bipartite walk $\hat{W}=\hat{R}_Y'\hat{R}_X\in \mathcal{B}$ on a connected bipartite multigraph $G=(X\sqcup Y, E)$, 
let the unitary map $\mathcal{U}_{\xi_Z}: \ell^2(E)\to \ell^2(A_Z)$ be as above described $(Z=X,Y)$.
Then, there exists a coined walk $\hat{U}\in \mathcal{C}$ in $\ell^2(A)$ such that 
	\[ \hat{W}=\mathcal{U}_{\xi_X}^{-1}\hat{U}^2\mathcal{U}_{\xi_X}, \]
where $\hat{U}$ is the evolution operator of a coined quantum walk on $\ell^2(A)$ so that $\hat{U}=\hat{S}\hat{C}$ 
with $\hat{C}= \mathcal{U}_{\xi_X} \hat{R}_X \mathcal{U}_{\xi_X}^{-1} \oplus \mathcal{U}_{\xi_Y} \hat{R}_Y \mathcal{U}_{\xi_Y}^{-1}$ under the decomposition 
$\ell^2(A)=\ell^2(A_X) \oplus \ell^2(A_Y)$. 
\end{proposition} 
\begin{proof}
First we show that 
$\mathcal{U}_{\xi_X} \hat{R}_x\mathcal{U}_{\xi_X}^{-1}$ and 
$\mathcal{U}_{\xi_X} \hat{R}_y'\mathcal{U}_{\xi_X}^{-1}$ 
($x\in X$, $y\in Y$) are 
local operators on $\{\delta_a \;|\; t(a)=x\}$, $\{\delta_b \;|\; o(b)=y\}\subset \ell^2(A)$, respectively. 
It holds
	\[
         t(a)= x \Leftrightarrow t(\xi_X(e))=x \Leftrightarrow X(e)=x,
         \]
where we put $a=\xi_X(e)$ and $a'=\xi_X(e')$. 
Using this, we have 
	\begin{align*}
        t(a)\neq x \;\mathrm{or}\; t(a')\neq x 
        	& \Leftrightarrow X(e)\neq x \;\mathrm{or}\; X(e')\neq x \\
                & \Rightarrow \langle \delta_e,\hat{R}_x\delta_{e'} \rangle=0 \\
                & \Leftrightarrow \langle \delta_a,\mathcal{U}_{\xi_X} \hat{R}_x\mathcal{U}_{\xi_X}^{-1}\delta_{a'}\rangle=0. 
        \end{align*}
Thus $\mathcal{U}_{\xi_X} \hat{R}_x\mathcal{U}_{\xi_X}^{-1}$ is a local operator on $\spann\{\delta_a \;|\; t(a)=x\}$. 
In the same way, 
it holds
	\[
         o(b)= y \Leftrightarrow o(\xi_X(e))=y \Leftrightarrow Y(e)=y,
         \]
where we put $b=\xi_X(e)$ and $b'=\xi_X(e')$. 
Using this, we have 
	\begin{align*}
        o(b)\neq y \;\mathrm{or}\; o(b')\neq y 
        	& \Leftrightarrow Y(e)\neq y \;\mathrm{or}\; Y(e')\neq y \\
                & \Rightarrow \langle \delta_e,\hat{R}_y'\delta_{e'} \rangle=0 \\
                & \Leftrightarrow \langle \delta_b,\mathcal{U}_{\xi_X} \hat{R}'_x\mathcal{U}_{\xi_X}^{-1}\delta_{b'}\rangle=0. 
        \end{align*}
Thus $\mathcal{U}_{\xi_X} \hat{R}_y'\mathcal{U}_{\xi_X}^{-1}$ is a local operator on $\spann\{\delta_b \;|\; o(b)=y\}$. 
On the other hand,  in a similar fashion, we can also show that 
$\mathcal{U}_{\xi_Y} \hat{R}_y'\mathcal{U}_{\xi_Y}^{-1}$ and $\mathcal{U}_{\xi_Y} \hat{R}_x\mathcal{U}_{\xi_Y}^{-1}$ are 
local operators on $\{\delta_a \;|\; t(a)=y\}$, $\{\delta_a \;|\; o(a)=x\}\subset \ell^2(A)$, respectively.  
By Lemma~\ref{lem}, setting 
	\begin{align*}
        \hat{E} &:= \mathcal{U}_{\xi_X}\hat{R}_X\mathcal{U}^{-1}_{\xi_X} \oplus \mathcal{U}_{\xi_Y}\hat{R}_Y'\mathcal{U}^{-1}_{\xi_Y}, \\
        \hat{F} &:= \mathcal{U}_{\xi_X}\hat{R}_Y'\mathcal{U}^{-1}_{\xi_X} \oplus \mathcal{U}_{\xi_Y}\hat{R}_X\mathcal{U}^{-1}_{\xi_Y}
        \end{align*}
under the decomposition $\ell^2(A)=\ell^2(A_X)\oplus \ell^2(A_Y)$, we see that $\hat{F}\hat{E}:\ell^2(A)\to\ell^2(A)$ describes an evolution operator of a
two-step coined walk $\hat{U}^2$ on $G$. Therefore
	\begin{align*}
        \hat{U}^2 
        	&= \mathcal{U}_{\xi_X}\hat{R}_Y\hat{R}_X\mathcal{U}_{\xi_X}^{-1} \oplus \mathcal{U}_{\xi_Y}\hat{R}_X\hat{R}_Y\mathcal{U}_{\xi_Y}^{-1} 
        \end{align*}
Putting $\Pi_{\ell^2(A_X)}$ as the projection onto $\ell^2(A)_X$, we have 
	\begin{align*}
        &\hat{U}^2\Pi_{\ell^2(A_X)}
        	= \mathcal{U}_{\xi_X}\hat{R}_Y\hat{R}_X\mathcal{U}_{\xi_X}^{-1} \\
       \Leftrightarrow\; & 
       		\hat{U}^2\Pi_{\ell^2(A_X)}\mathcal{U}_{\xi_X}
        	   = \mathcal{U}_{\xi_X} \hat{R}_Y\hat{R}_X \\
       \Leftrightarrow\; & 
       		\hat{U}^2\mathcal{U}_{\xi_X}
        	   = \mathcal{U}_{\xi_X} \hat{R}_Y\hat{R}_X \\        
       \Leftrightarrow\; & 
       		\mathcal{U}_{\xi_X}^{-1}\hat{U}^2\mathcal{U}_{\xi_X}
        	   = \hat{R}_Y\hat{R}_X=\hat{W}. 
        \end{align*}
Thus, we obtain the desired conclusion.
\end{proof}
\subsubsection{Proof of $\mathcal{B} \succ \mathcal{C}_2$}
Let $G_2=(V_2,E_2)$ be the duplicated multigraph of $G=(V,E)$. 
We call the bijection map from $V\to V'$ by $\phi$, where $V'$ is the copy of $V$, that is, $\phi(v)=v'$ and $\phi^{-1}(v')=v$. 
The end vertex in $V$ is denoted by $V(e)$, and one in $V'$ is denoted by $V'(e)$ for $e\in E_2$. 
The symmetric arc set of $G$ is denoted by $A$. 
The central players are $E_2$ and $A$, and the bijection map $\eta: E_2\to A$ is defined by 
	\[ t(\eta(e))=V(e),\;o(\eta(e))=\phi^{-1}(V'(e)). \]
The inverse map is 
	\[ V(\eta^{-1}(a))=t(a),\;V'(\eta^{-1}(a))=\phi(o(a)).  \]
This is equivalent to that $u$ and $\phi(v)$ is adjacent in $G_2$ if and only if there exists an arc  $a$ such that $t(a)=u$ and $o(a)=v$ in $G$. 
Note that $\eta^{-1}(\bar{a})$ and $\eta^{-1}(a)$ give the following crossing relation:
	\[ \phi(V(\eta^{-1}(a))) = V'(\eta^{-1}(\bar{a})),\;\phi(V(\eta^{-1}(\bar{a}))) = V'(\eta^{-1}(a)). \] 
The unitary map induced by $\eta$, $\mathcal{U}_\eta: \ell^2(E_2) \to \ell^2(A)$, is 
	\[ (\mathcal{U}_\eta\psi)(a)=\psi(\eta^{-1}(a)). \]
Using the bijection map $\eta$, we obtain the following proposition which implies $\mathcal{B} \succ \mathcal{C}_2$.
\begin{figure}[htbp]
  \begin{center}
            \includegraphics[clip, width=5cm]{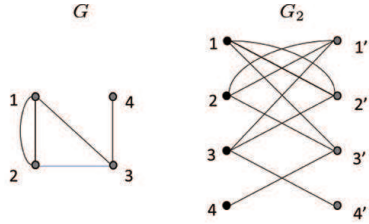}
              \caption{Duplicated multigraph.} 
  \end{center}
\end{figure}
\begin{proposition}
Let $\hat{U}=\hat{S}\hat{C}$ be the coined walk on $G$ with $\hat{C}=\oplus_{u\in V}\{\hat{C}_u\}$. There exists a bipartite walk 
$\hat{W}=\left(\oplus_{v'\in V'}\hat{R}'_{v'}\right)\left(\oplus_{v\in V}\hat{R}_{v}\right)$ on $G_2$ such that 
	\[ \hat{U}^2=\mathcal{U}_\eta \hat{W} \;\mathcal{U}_\eta^{-1}. \]
The local unitary operators of the bipartite walk are  
	\begin{equation}\label{pipi} 
	\hat{R}_u=\mathcal{U}_\eta^{-1}\hat{C}_u\mathcal{U}_\eta,\; \hat{R}_u'=\mathcal{U}_\eta^{-1}\hat{S}\hat{C}_u\hat{S}\mathcal{U}_\eta. 
	\end{equation}
\end{proposition}
\begin{proof}
We show $\mathcal{U}_\eta^{-1} \hat{U}^2 \;\mathcal{U}_\eta$ is an evolution operator of a bipartite walk on $G_2$. 
By Lemma~\ref{lem}, the two-step coined walk on $G$ is expressed by $\hat{F}\, \hat{C}$, where $\hat{C}$ and $\hat{F}$ are 
direct sums of $\{ \hat{C}_u \}_{u\in V}$ and $\{\hat{F}_u\}_{u\in V}$ following the decompositions of arcset $A$; 
$\sqcup_{u\in V}\{a\in A \;|\; t(a)=u\}$ and $\sqcup_{u\in V}\{a\in A \;|\; o(a)=u\}$, respectively. 
Here $\hat{F}_u=S\hat{C}_uS$ for every $u\in V$. 
First we need to show that $\mathcal{U}_\eta^{-1} \hat{C}_u\mathcal{U}_\eta$ and $\mathcal{U}_\eta^{-1} \hat{F}_u\mathcal{U}_\eta$ are 
local unitary operators on $\spann\{ \delta_{e} \;|\; V(e)=u \}$ and $\spann\{ \delta_{e} \;|\; V'(e)=u' \}$, where $u'$ is the copy of $u$. 
We put $\eta^{-1}(a)=e$ and $\eta^{-1}(b)=f$. 
For $u\in V$, it holds 
	\begin{align*}
        V(e)=u \Leftrightarrow V(\eta^{-1}(a))=u \Leftrightarrow t(a)=u. 
        \end{align*}
Using this, we have 
	\begin{align*}
        V(e)\neq u \;\mathrm{or}\; V(f)\neq u
        	& \Leftrightarrow t(a)\neq u \;\mathrm{or}\; t(b)\neq u \\
                & \Rightarrow \langle \delta_a,\hat{C}_u \delta_b\rangle=0 \\
                & \Leftrightarrow \langle \delta_e,\mathcal{U}_{\eta}^{-1}\hat{C}_u \mathcal{U}_{\eta}\delta_{f}\rangle=0.
        \end{align*}
In the same way, for $u'\in V'$, it holds 
	\begin{align*}
        V'(e)=u' \Leftrightarrow V'(\eta^{-1}(a))=u' \Leftrightarrow o(a)=u. 
        \end{align*}
Using this, we have 
	\begin{align*}
        V'(e)\neq u' \;\mathrm{or}\; V'(f)\neq u'
        	& \Leftrightarrow o(a)\neq u \;\mathrm{or}\; o(b)\neq u \\
                & \Rightarrow \langle \delta_a,\hat{F}_u \delta_b\rangle=0 \\
                & \Leftrightarrow \langle \delta_e,\mathcal{U}_{\eta}^{-1}\hat{F}_u \mathcal{U}_{\eta}\delta_f\rangle=0.
        \end{align*}
Then, $\mathcal{U}_\eta^{-1} \hat{C}_u\mathcal{U}_\eta$ and $\mathcal{U}_\eta^{-1} \hat{F}_u\mathcal{U}_\eta$ are 
local unitary operators on $\spann\{ \delta_{e} \;|\; V(e)=u \}$ and $\spann\{ \delta_{e} \;|\; V'(e)=u' \}$ for every $u\in V$. 
Therefore, by Lemma~\ref{lem}, $\mathcal{U}_\eta^{-1} \hat{F}\hat{C}\mathcal{U}_\eta$ is the evolution operator of a bipartite walk on 
$G_2$. 
This completes the proof. 
\end{proof}
In the rest of this section, we consider a special bipartite multigraph which is a duplicated multigraph. 
\begin{lemma}\label{lem2}
If a bipartite multigraph is the duplicated multigraph of $H$ and the evolution operator of the bipartite walk $\hat{W}$ on this bipartite multigraph is given by 
$\{\hat{R}_v\}_{v\in V}$ and $\{\hat{R}'_{v'}\}_{v'\in V'}$ satisfying
	\[ \langle\delta_e, \hat{R}_v\delta_{f}\rangle = \langle\delta_{e'}, \hat{R}_{v'}'\delta_{f'}\rangle, \]
where $(V(e'))'=V'(e)$ and $(V(f'))'=V'(f)$, then $\hat{W}$ is unitarily equivalent to the two-step coined quantum walk $\hat{U}$ on $H$ as follows:
	\[ \hat{W}=\mathcal{U}_\eta^{-1}\hat{U}^2\mathcal{U}_\eta. \]
Here the local coin operators of $\hat{U}$ are described by $\{\mathcal{U}_\eta\hat{R}_u\mathcal{U}_\eta^{-1}\}_{u\in V}$. 
\end{lemma}
\begin{proof}
Similar to the previous proofs, it is easy to show that 
$\mathcal{U}_\eta\hat{R}_v\mathcal{U}_\eta^{-1}$ and $\mathcal{U}_\eta\hat{R}'_{v'}\mathcal{U}_\eta^{-1}$ are local unitary operators on
$\spann\{\delta_a \;|\; t(a)=v\}$ and $\spann\{\delta_a \;|\; o(a)=v\}$, respectively $(v\in V,\;v'\in V')$. 
Moreover for $e,e',f,f'\in E_2$ with $(V(e'))'=V'(e)$ and $(V(f'))'=V'(f)$, the condition 
$\langle \delta_f, \hat{R}_v\delta_e \rangle = \langle \delta_{f'}, \hat{R}_{v'}'\delta_{e'} \rangle $ is equivalent to 
	\[ \langle \mathcal{U}_{\eta}^{-1}\delta_{\eta(f)}, \hat{R}_v \mathcal{U}_\eta^{-1}\delta_{\eta(e)} \rangle 
           = \langle \mathcal{U}_{\eta}^{-1}\delta_{\eta(f')}, \hat{R}_v'\mathcal{U}_\eta^{-1}\delta_{\eta(e')} \rangle. \]
Note that $\eta(e')=\overline{\eta(e)},\;\eta(f')=\overline{\eta(f)}\in A$ hold. Thus putting $a=\eta(e)$ and $b=\eta(f)$, we have 
	\begin{align*}
        & \langle \delta_f, \hat{R}_v\delta_e \rangle = \langle \delta_{f'}, \hat{R}_{v'}'\delta_{e'} \rangle\\
        \Leftrightarrow \; 
        & \langle \mathcal{U}_{\eta}^{-1}\delta_a, \hat{R}_v \mathcal{U}_\eta^{-1}\delta_{b} \rangle 
        = \langle \mathcal{U}_{\eta}^{-1}\delta_{\bar{a}}, \hat{R}_{v'}'\mathcal{U}_\eta^{-1}\delta_{\bar{b}} \rangle \\
        \Leftrightarrow \;
	& \langle \delta_a, \mathcal{U}_{\eta}\hat{R}_v \mathcal{U}_\eta^{-1}\delta_{b} \rangle 
        = \langle \delta_{\bar{a}}, \mathcal{U}_{\eta}\hat{R}_{v'}'\mathcal{U}_\eta^{-1}\delta_{\bar{b}} \rangle \\
        \Leftrightarrow \;
	& \langle \delta_a, \mathcal{U}_{\eta}\hat{R}_v \mathcal{U}_\eta^{-1}\delta_{b} \rangle 
        = \langle \delta_{a}, S\mathcal{U}_{\eta}\hat{R}_{v'}'\mathcal{U}_\eta^{-1}S\delta_{b} \rangle.
        \end{align*}
Therefore, we have shown that $\mathcal{U}_\eta^{-1}\hat{W}\mathcal{U}_\eta$ describes a 2-step coined walk on $H$. 
\end{proof}
Lemma~\ref{lem2} leads to the following corollary:
\begin{corollary}\label{UEQS}
Let $G=(V,E)$ be a connected multigraph and $G_2=(V\sqcup V',E_2)$ be its duplicated multigraph, where $V'$ is the copy of $V$. 
The set of symmetric arcs of $G$ is denoted by $A$. 
The quantum search driven by a bipartite walk on $G_2$ with respect to (\ref{revisedSzeMethod1}) and (\ref{revisedSzeMethod2}) and 
the square of one driven by coined walk on $G$ for case I are unitary equivalent with respect to a unitary map $\mathcal{U}_{\eta}: \ell^2(E_2)\to \ell^2(A)$.
Here the unitary map $\mathcal{U}_\eta$ is denoted as follows: 
\[ (\mathcal{U}_\eta\psi)(a)=\psi(\eta^{-1}(a)), \]
where the bijection map $\eta: E_2\to A$ is 
	\[ t(\eta(e))=V(e),\;o(\eta(e))=\phi^{-1}(V'(e)). \]
\end{corollary}
%

\section{Spectral analysis of coined walks}
As discussed in the above section, the quantum walks analyzed in this work can be interpreted as a two-step coined walk. We put our attention in the class of coined walks in order to analyze it in more detail. The total Hilbert space in this case is $\mathcal{H}=\ell^2(A)$. 
Now we show the spectral map theorem of coined walks with some special coin. 
\subsection{Setting}
For given connected graph G=(V,A), we assign local unitary operators $\hat{C}_u$ for each $u\in V$ under the decomposition 
$\mathcal{H}:=\ell^2(A)=\oplus_{u\in V}\{\psi \;|\; t(a)\neq u \Rightarrow \psi(a)=0\}$ in the coined walk. 
We assume $\sigma(\hat{C}_u) \subseteq \{\pm 1\}$, where $\sigma(\cdot)$ is the spectrum. 
The subspace $\mathcal{C}_u$ are decomposed into
	\[ \mathcal{C}_u=\ker(1-\hat{C}_u) \oplus \ker(1+\hat{C}_u). \]
\begin{remark}
This setting includes all previous examples for quantum searches of $M$, that is, \\

\[ \mathrm{Case\;I:\;\;} \dim \ker(1-\hat{C}_u) = \begin{cases} 1 & \text{if $u\notin M$,}\\ 0 & \text{if $u\in M$.} \end{cases} \]
\[ \mathrm{Case\;II:\;\;} \dim \ker(1-\hat{C}_u) = \begin{cases} 1 & \text{if $u\notin M$,}\\ \mathrm{deg}(u)-1 & \text{if $u\in M$.} \end{cases} \]
\end{remark}
Putting $d_u=\dim\ker(1-\hat{C}_u)$, 
we set 
	\begin{equation}\label{newV} \tilde{V}:=\{(u,\ell) \;|\; \ker(1-\hat{C}_u)\neq \{\bs{0}\},\; \ell=1,\dots,d_u\}. \end{equation}
We define $\mathcal{K}=\ell^2(\tilde{V})$ such that 
	\begin{align*} 
        \mathcal{K} &:= \bigoplus_{u: \ker(1-\hat{C}_u)\neq \bs{0}}\mathbb{C}^{d_u}=\spann\{\;|u;\ell \rangle \;|\; (u,\ell)\in \tilde{V}\}. 
        \end{align*}
Here $|u;\ell\rangle$ denotes the standard basis of $\mathcal{V}$. 
We set $\hat{U}=\hat{S}\hat{C}$ where $S$ is the flip-flop shift operator and $C=\oplus_{u\in V}\hat{C}_u$. 
We will express the spectrum of $U$ on $\ell^2(A)$ whose cardinality is $|A|$ by some self-adjoint operator on $\mathcal{K}$ 
whose cardinality is reduced to $|\tilde{V}| \leq |A|$. 
\subsection{Boundary operator}
Let the complete orthogonal normalized system (CONS) of $\ker(1-\hat{C}_u)\neq \{\bs{0}\}$ be $\{\alpha_u^{(\ell)}\}_{\ell=1}^{d_u}$. 
We define $\partial: \mathcal{H}\to \mathcal{K}$ by 
	\begin{equation}\label{d_1*} 
        (\partial\psi)(u;\ell)=\langle \alpha_u^{(\ell)},\psi \rangle. 
        \end{equation}
It is equivalent to 
	\[ \partial \delta_a 
        	= 
                \sum_{\ell=1}^{d_{t(a)}} \overline{\alpha_{t(a)}^{(\ell)}(a)}\;|t(a);\ell\rangle. 
        \]
The adjoint operator 
$\partial^*: \mathcal{K} \to \mathcal{H}$ is given by 
	\[ (\partial^* f)(a)=\sum_{\ell=1}^{d_{t(a)}} f(t(a);\ell)\alpha_{t(a)}^{(\ell)}(a). \]
We observe that
	\[ (\partial^* f)(a)=\left\langle \overline{f(t(a);\cdot)},\;\alpha_{t(a)}^{(\cdot)}(a) \right\rangle_{\mathbb{C}^{d_{t(a)}}}. \]
It is equivalent to 
	\begin{equation}\label{d_1} 
        \partial^*|u;\ell\rangle = \alpha_u^{(\ell)}. 
        \end{equation}

The following important relations hold: 
	\begin{align}
        \partial \partial^* &= \bs{1}_{\mathcal{K}}; \\
        \partial^* \partial &= \Pi_{\oplus_{u\in V}\ker(1-C_u)},
        \end{align}
where $\Pi_{\mathcal{H}'}$ is the projection onto $\mathcal{H}'\subset \ell^2(A)$. 
Therefore the coin operator $\hat{C}$ is expressed by 
	\begin{equation}
        \hat{C}=2\partial^*\partial-\bs{1}_{\mathcal{H}}.
        \end{equation}
\subsection{Underlying graph and a dynamics on it}
%
\begin{definition}
Let $G=(V,A)$ be the symmetric directed graph which may have multiple arcs 
and $\tilde{V}$ be defined by (\ref{newV}) induced by $V$ and $\{d_u\}_{u\in V}$. 
We set the underlying graph $\tilde{G}$ determined by $(G,\{d_u\}_{u\in V})$ as follows.
The set of vertices of $\tilde{G}$ is $\tilde{V}$. 
The set of symmetric arcs $\tilde{A}$ of $\tilde{G}$ is given by
	\[ \#\{ \tilde{a}\in \tilde{A} \;|\; t(\tilde{a})=(u;\ell),\;o(\tilde{a})=(u',\ell') \}=\#\{ a\in A \;|\; t(a)=u,\;o(a)=u' \} \]
for every $\ell=1,\dots,d_u$ and $\ell'=1,\dots,d_{u'}$. 
\end{definition}
We define a weight $w: \tilde{A}\to \mathbb{C}$ by 
	\[ w(\tilde{a})=\alpha_{t(a)}^{(\ell)}(a) \]
for $\tilde{a}\in \tilde{A}$ such that $o(\tilde{a})=(o(a);\ell)$ and $t(\tilde{a})=(t(a);\ell)$. 
\begin{definition}
The operator $\hat{T}: \mathcal{K} \to \mathcal{K}$ is defined by 
	\[ (\hat{T}f)(\tilde{u})=\sum_{b\in \tilde{A}\;:\;t(b)=\tilde{u}}\overline{w(b)}w\left(\bar{b}\right)f(o(b)) \] 
for every $\tilde{u}\in\tilde{V}$ and $f\in \mathcal{K}$. 
\end{definition}
\begin{lemma}
	\[ \hat{T}=\partial S\partial^*; \]
	\[ \sigma(\hat{T})\subseteq [-1,1]. \]
\end{lemma}
\begin{proof}
The first part is obtained by a direct computation. 
For the second part of the proof, put $\mu\in \sigma(\hat{T})$ and $f\in \ker (\mu-\hat{T})$. Then 
	\begin{align*} 
        |\mu|^2 ||f||^2 &= \langle \partial \hat{S}\partial^* f, \partial \hat{S}\partial^* f \rangle \\
        		&= \langle \hat{S}f, \Pi_1 \hat{S}\partial f  \rangle 
                        \leq  \langle \partial \hat{S}f, \partial \hat{S}f  \rangle \\
                        &= \langle f, f  \rangle
                        \leq ||f||^2,
        \end{align*}
where $\Pi_1:=\partial^*\partial$. 
\end{proof}
\subsection{Spectrum of $U^2\in \mathcal{C}_2$}
We set $\mathcal{L}:=\partial^* \mathcal{K}+\hat{S}\partial^*\mathcal{K} \subset \mathcal{H}$, which is called the inherited subspace. 
In \cite{MOS}, $\sigma(\hat{C}_u)=\{\pm 1\}$ and $\dim \ker(1-\hat{C}_u)=1$ for any $u\in V$ were assumed, on the other hand, 
we relax this assumption to $\sigma(\hat{C}_u) \subseteq \{\pm 1\}$; 
the eigenvalues and its multiplicities of $C_u$ depend on $u\in V$. 
However a similar argument to \cite{MOS} holds and the proof is essentially same as \cite{MOS}. Thus we skip its proof. 
\begin{theorem}
Let $G=(V,A)$ be a connected multigraph. The unitary operator $\hat{U}$ on $\mathcal{H}$ denotes 
the evolution operator of a coined quantum walk on $G$ with the coin operator $\{\hat{C}_u\}_{u\in V}$, where $\sigma(\hat{C}_u)\subseteq \{\pm 1\}$. 
The evolution operator of the underlying cellular automaton on $\tilde{G}$ is denoted by $T: \mathcal{K}\to \mathcal{K}$.  
Then we have 
\[ \hat{U}=\hat{U}_{\mathcal{L}}\oplus \hat{U}_{\mathcal{L}^\perp} \]
and 
\begin{multline*}
\ker(e^{i\theta}-\hat{U})= \\
\begin{cases}
	\{ \frac{1}{\sqrt{2}|\sin\theta|}\;(1-e^{i\theta} \hat{S})\partial^*f_{\cos\theta} \;|\; f_{\cos \theta} \in \ker(\cos\theta-\hat{T}) \} & \text{: $e^{i\theta}\in \sigma(\hat{U})\setminus\{\pm 1\}$,} \\
        \{ \partial^* f_{\cos \theta} \;|\; f_{\cos \theta} \in \ker(\cos\theta-\hat{T}) \} & \text{: $e^{i\theta}\in \sigma(\hat{U}|_{\mathcal{L}})\cap\{\pm 1\}$,} \\
	\ker(\partial) \cap \ker(1\pm \hat{S}) & \text{: $e^{i\theta}\in \sigma(\hat{U}|_{\mathcal{L}^\perp})\cap\{\pm 1\}$.}
\end{cases}
\end{multline*}
\end{theorem}
The above theorem immediately leads to the following corollary: 
\begin{corollary}
The setting and notations are same as begore. Then, 
\begin{multline*}
\ker(e^{i\theta}-\hat{U}^2)= \\
\begin{cases}
	\{ \frac{1}{\sqrt{2}|\sin(\theta/2)|}\;(1-e^{i\theta/2} \hat{S})\partial^*f_{\cos(\theta/2)} \;|\; f_{\cos (\theta/2)} \in \ker[\cos(\theta/2)-\hat{T}] \} & \text{: $e^{i\theta}\in \sigma(\hat{U}^2)\setminus\{1\}$,} \\
        \{ \partial^* f_{\cos (\theta/2)} \;|\; f_{\cos (\theta/2)} \in \ker[\cos(\theta/2)-\hat{T}] \} & \text{: $e^{i\theta}\in \sigma(\hat{U}^2|_{\mathcal{L}})\cap\{1\}$,} \\
	\ker(\partial) \cap \ker(\partial \hat{S}) & \text{: $e^{i\theta}\in \sigma(\hat{U}^2|_{\mathcal{L}^\perp})\cap\{1\}$.}
\end{cases}
\end{multline*}
\end{corollary}
Let $G=(X\sqcup Y,A)$ be a bipartite graph. We set 
\[ \mathcal{A}_X:=\{ \psi\in \ell^2(A) \;|\; t(a)\notin X \Rightarrow \psi(a)=0 \}, \] 
and 
\[ \mathcal{A}_Y:=\{ \psi\in \ell^2(A) \;|\; t(a)\notin Y \Rightarrow \psi(a)=0 \}. \]
We put $\hat{U}_{YX}:=\hat{U}\Pi_{\mathcal{A}_X}$, $\hat{U}_{XY}=\hat{U}\Pi_{\mathcal{A}_Y}$. Since $G$ is a bipartite graph, we have 
\[ (\hat{U}\Pi_{\mathcal{A}_X})^2 = \hat{U}_{XY}\hat{U}_{YX}. \]

If we are given a two-partition walk $\hat{U}'\in \mathcal{P}$, then by Theorem~\ref{unitaryeq}, we can convert this walk on $(\Omega;\pi_1,\pi_2)$ 
to some two-step coined walk $\hat{U}^2\in \mathcal{C}_2$. 
This walk is a coined walk on some bipartite graph $G=(X\sqcup Y,A)$ which is an intersection graph 
and $\hat{U}'$ is unitary equivalent to $\mathcal{}U_{XY}U_{YX}$ with a unitary map $\mathcal{W}:\ell^2(\Omega)\to \ell^2(A)$. 
By using the commutative diagram in Fig.~\ref{commutativemap}, the unitary map $\mathcal{W}$ is expressed as follows: 
	\[ (\mathcal{W}\psi)(a)=\psi(\gamma_E^{-1}\xi_X^{-1}(a)). \]
Then, the spectral analysis of $\hat{U}'$ is essentially obtained by the following corollary: 
\begin{corollary}
Let $G=(X\sqcup Y,A)$ be the intersection multigraph induced by $(\Omega;\pi_1,\pi_2)$, 
and $\hat{U}'\in \mathcal{P}$ be an evolution operator on $(\Omega;\pi_1,\pi_2)$. 
Moreover, let 
\[ \hat{U} := \begin{bmatrix} 0 & \hat{U}_{YX} \\ \hat{U}_{XY} & 0 \end{bmatrix}\in \mathcal{C}\] be the unitary operator equivalent 
to $\hat{U}'$ on $(X\sqcup Y,A)$. 
Then, we have 
\begin{align*}
\ker(e^{i\theta}-\hat{U}')=
\mathcal{W}^{-1}\ker(e^{i\theta}-\hat{U}_{XY}\hat{U}_{YX})\mathcal{W}= \mathcal{W}^{-1}\Pi_{\mathcal{A}_X}\ker(e^{i\theta}-\hat{U}^2)\mathcal{W}.
\end{align*}
\end{corollary}
%
%

\section{Conclusions}

\rp{
From the pure-mathematics viewpoint, the quantum walk is a strikingly interesting area due to the richness of the models. The main result of this work is Theorem~\ref{unitaryeq}, which shows that there are four families of quantum walk models unitarily equivalent, namely, (1)~the two-step coined model, (2)~the extension of Szegedy's model for multigraphs, (3)~the two-tessellable staggered model, and (4)~the two-partition model. The details on the equivalence between families~(1) and~(2) was addressed in Lemma~\ref{lem}. The family of coined quantum walks (one-step coined model) is strictly included in all those models listed above, that is, none of the above families is included in the coined model. Notice that the only demand in the coin choice in the coined model is unitarity, that is, no explicit formula for the coin operator is imposed. The locality is fulfilled because the coin space in an internal space. The same kind of unitary freedom must be allowed to the Szegedy model on multigraphs and to the two-tessellable staggered model provided the locality is fulfilled.

As a future work, it is interesting to analyze $k$-tessellable staggered models with $k>2$.
}

\par
\
\par\noindent
\noindent
{\bf Acknowledgments.}
\par
NK is partially supported by the Grant-in-Aid for Scientific Research (Chal-lenging Exploratory Research) of Japan Society for the Promotion of Science (Grant No. 15K13443). 
RP acknowledges financial support from Faperj (Grant No. E-26/102.350/2013) and CNPq (Grant No. 303406/2015-1).
IS is partially supported by the Grant-in-Aid for Scientific Research (C) of Japan Society for the Promotion of Science (Grant No.~15K04985). 
ES acknowledges financial support from 
the Grant-in-Aid for Young Scientists (B) and of Scientific Research (B) Japan Society for the Promotion of Science (Grant No.~16K17637, No.~16K03939).
\par

\begin{small}

\end{small}

\end{document}